\documentclass[onecolumn,12pt]{IEEEtran} 
\usepackage{color} 
\usepackage{amsfonts,latexsym}
\usepackage{amssymb}
\newcommand{\ord}{{\rm ord}}

\newcommand{\tr}{{\rm Tr}}
\newcommand{\gf}{{\rm GF}}

\newcommand{\calC}{{\mathcal C}}
\newcommand{\C}{{\mathcal C}}
\newcommand{\ls}{{\mathbb L}}

\newcommand{\bc}{{\mathbf{c}}}

\newtheorem{theorem}{Theorem}[section]
\newtheorem{lemma}[theorem]{Lemma}

\newtheorem{example}[theorem]{Example}

\newtheorem{remark}[theorem]{Remark} 

\newtheorem{open}[theorem]{Open Problem}

\begin{document}

\title{Cyclic Codes from Cyclotomic Sequences of Order Four}

\author{Cunsheng Ding\thanks{
C. Ding is with the Department of Computer Science and Engineering, 
The Hong Kong University of Science and Technology, Clear Water Bay, 
Kowloon, Hong Kong. His research is supported by the Hong Kong Research 
Grants Council under Grant No. 601311. Email: cding@ust.hk}}

\date{\today}
\maketitle

\begin{abstract} 
Cyclic codes are an interesting subclass of linear codes and have been used 
in consumer electronics, data transmission technologies, broadcast systems, 
and computer applications due to their efficient encoding and decoding 
algorithms. In this paper, three cyclotomic sequences of order four are 
employed to construct a number of classes of cyclic codes over $\gf(q)$ 
with prime length. Under certain conditions lower bounds on the minimum 
weight are developed. Some of the codes obtained are optimal or almost 
optimal. In general, the cyclic codes constructed in this paper are very good. 
Some of the cyclic codes obtained in this paper are closely related to almost
difference sets and difference sets. As a byproduct, the $p$-rank of these 
(almost) difference sets are computed.    
\end{abstract}

\begin{keywords} 
Almost difference sets, cyclic codes, cyclotomy, difference sets, sequences  
\end{keywords}

\section{Introduction}

Let $q$ be a power of a prime $p$. 
A linear $[n,k, d]$ code over $\gf(q)$ is a $k$-dimensional subspace of $\gf(q)^n$ 
with minimum (Hamming) nonzero weight $d$. 

A linear $[n,k]$ code $\C$ over the finite field $\gf(q)$ is called {\em cyclic} if 
$(c_0,c_1, \cdots, c_{n-1}) \in \C$ implies $(c_{n-1}, c_0, c_1, \cdots, c_{n-2}) 
\in \C$.  
Let $\gcd(n, q)=1$. By identifying any vector $(c_0,c_1, \cdots, c_{n-1}) \in \gf(q)^n$ 
with  
$$ 
c_0+c_1x+c_2x^2+ \cdots + c_{n-1}x^{n-1} \in \gf(q)[x]/(x^n-1), 
$$
any code $\C$ of length $n$ over $\gf(q)$ corresponds to a subset of $\gf(q)[x]/(x^n-1)$. 
The linear code $\C$ is cyclic if and only if the corresponding subset in $\gf(q)[x]/(x^n-1)$ 
is an ideal of the ring $\gf(q)[x]/(x^n-1)$. 

Note that every ideal of $\gf(q)[x]/(x^n-1)$ is principal. Let $\C=(g(x))$ be a 
cyclic code. Then $g(x)$ is called the {\em generator polynomial} and 
$h(x)=(x^n-1)/g(x)$ is referred to as the {\em parity-check} polynomial of 
$\C$. 

A vector $(c_0, c_1, \cdots, c_{n-1}) \in \gf(q)^n$ is said to be {\em even-like} 
if $\sum_{i=0}^{n-1} c_i =0$, and is {\em odd-like} otherwise. The minimum weight 
of the even-like codewords, respectively the odd-like codewords of a code is the 
minimum even-like weight, denoted by $d_{even}$, respectively the minimum 
odd-like weight of the code, denoted by $d_{odd}$.  

The error correcting capability of cyclic codes may not be as good as some other linear 
codes in general. However, cyclic codes have wide applications in storage and communication 
systems because they have efficient encoding and decoding algorithms 
\cite{Chie,Forn,Pran}.

Cyclic codes have been studied for decades and a lot of  progress has been made 
(see for example, \cite{BS06,Char,EL,HPbook,LintW}).  The total number of cyclic 
codes over $\gf(q)$ and their constructions are closely related to cyclotomic cosets 
modulo $n$, and thus many areas of number theory. An important problem in 
coding theory is to find simple ways to construct good cyclic codes. 

In this paper, we construct cyclic codes over $\gf(q)$ with length $n$ and generator 
polynomial 
\begin{eqnarray}\label{eqn-defseqcode}
\frac{x^n-1}{\gcd(\Lambda(x), x^n-1)}
\end{eqnarray}
where 
$$ 
\Lambda(x)=\sum_{i=0}^{n-1} \lambda_i x^i  \in \gf(q)[x]   
$$
and $\lambda^{\infty}=(\lambda_i)_{i=0}^{\infty}$ is a sequence of period $n$ over $\gf(q)$. 
Throughout this paper, we call the cyclic code $\C_\lambda$ with the generator polynomial 
of (\ref{eqn-defseqcode}) the {\em code defined by the sequence} $\lambda^{\infty}$, 
and the sequence $\lambda^{\infty}$ the {\em defining sequence} of the cyclic code 
$\C_\lambda$. By employing three cyclotomic sequences $\lambda^{\infty}$ over $\gf(q)$, 
we will construct several classes of cyclic codes over $\gf(q)$ with prime length. The cyclic 
codes presented in this paper are very good in general. Some of them are optimal or almost 
optimal. As a byproduct, the $p$-rank of some almost difference sets and difference sets 
are computed.  

\section{Preliminaries} 

In this section, we present basic notations and results of combinatorial designs,  
cyclotomy, sequences, and cyclic codes that will be employed in subsequent 
sections.     

\subsection{Difference sets and almost difference sets} 

Let $(A, +)$ be an Abelian group of order $n$. Let $C$ be a $k$-subset of 
$A$. The set $C$ is an $(n, k, \lambda)$ {\em difference set} of $A$ 
if $d_C(w)=\lambda$ for every nonzero element of $A$, where $d_C(w)$ is 
the {\em difference function} defined by 
$$ 
d_C(w)=|C \cap (C+w)|, 
$$ 
here and hereafter $C+w:=\{c+w: c \in C\}$.  Detailed information on 
difference sets can be found in \cite{BJL}. 

Let $(A, +)$ be an Abelian group of order $n$. A $k$-subset $C$ of 
$A$ is an $(n, k, \lambda, t)$ {\em almost difference set} of $A$ 
if $d_C(w)$ takes on $\lambda$ altogether $t$ times and  
$\lambda +1$ altogether $n-1-t$ times when $w$ ranges over all the nonzero 
elements of $A$.  The reader is referred to \cite{ADHKM} for information on 
almost difference sets. 

Difference sets and almost difference sets are closely related to sequences with 
only a few autocorrelation values, and are related to some of the codes constructed 
in this paper. 

\subsection{The linear span and minimal polynomial of periodic sequences}

Let $\lambda^n=\lambda_0\lambda_1\cdots \lambda_{n-1}$ be a sequence over $\gf(q)$. The {\em linear 
span} (also called {\em linear complexity}) of $\lambda^n$ is defined to be the smallest positive 
integer $\ell$ such that there are constants $c_0=1, c_1, \cdots, c_\ell \in \gf(q)$ 
satisfying 
\begin{eqnarray*} 
-c_0\lambda_i=c_1\lambda_{i-1}+c_2\lambda_{i-2}+\cdots +c_l\lambda_{i-\ell} \mbox{ for all } \ell \leq i<n. 
\end{eqnarray*} 
In engineering terms, such a polynomial $c(x)=c_0+c_1x+\cdots +c_lx^l$ 
is called the {\em feedback  polynomial} of a shortest linear feedback 
shift register 
(LFSR) that generates $\lambda^n$. Such an integer always exists for finite sequences  $\lambda^n$. When 
$n$ is $\infty$, a sequence $\lambda^{\infty}$ is called a semi-infinite 
sequence. If there is no such an integer for a semi-infinite sequence 
$\lambda^{\infty}$, its linear span is defined to be $\infty$. The linear 
span of the zero sequence is defined to be zero. 
For ultimately periodic semi-infinite sequences such an $\ell$ always 
exists. 

Let $\lambda^{\infty}$ be a sequence of period $n$ over $\gf(q)$. 
Any feedback polynomial of $\lambda^{\infty}$ is called an {\em characteristic 
polynomial}. The characteristic polynomial with the smallest degree is 
called the {\em minimal polynomial} of the periodic sequence $\lambda^{\infty}$. 
Since we require that the constant term of any characteristic polynomial 
be 1, the minimal polynomial of any periodic sequence $\lambda^{\infty}$ must 
be unique. In addition, any characteristic polynomial must be a multiple 
of the minimal polynomial.    

For periodic sequences, there are a few ways to determine their linear 
span and minimal polynomials. One of them is given in the following 
lemma \cite{LN97}. 

\begin{lemma}\label{lem-ls0} 
Let $\lambda^{\infty}$ be a sequence of period $n$ over $\gf(q)$. 
Define   
\begin{eqnarray*}
\Lambda^{n}(x)=\lambda_{0}+\lambda_{1}x+\cdots +\lambda_{n-1}x^{n-1} \in \gf(q)[x]. 
\end{eqnarray*}
Then the minimal polynomial $m_\lambda$ of $\lambda^{\infty}$ is given by 
      \begin{eqnarray}\label{eqn-base1}  
      \frac{x^{n}-1}{\gcd(x^{n}-1, \Lambda^{n}(x))};
      \end{eqnarray}  
and the linear span $\ls_\lambda$ of $\lambda^{\infty}$ is given by 
      \begin{eqnarray}\label{eqn-base2} 
       n-\deg(\gcd(x^{n}-1, \Lambda^{n}(x))). 
      \end{eqnarray}  
\end{lemma}

\subsection{Group characters and Gaussian sums}

Let $q$ be a power of a prime $p$. 
Let $\tr_{q/p}$ denote the trace function from $\gf(q)$ to $\gf(p)$.
An {\em additive character} of $\gf(q)$ is a nonzero function $\chi$
from $\gf(q)$ to the set of complex numbers such that
$\chi(x+y)=\chi(x) \chi(y)$ for any pair $(x, y) \in \gf(q)^2$.
For each $b\in \gf(q)$, the function
\begin{eqnarray}\label{dfn-add}
\chi_b(c)=e^{2\pi \sqrt{-1} \tr_{q/p}(bc)/p} \ \ \mbox{ for all }
c\in\gf(q)
\end{eqnarray}
defines an additive character of $\gf(q)$. When $b=0$,
$\chi_0(c)=1 \mbox{ for all } c\in\gf(q),
$
and is called the {\em trivial additive character} of
$\gf(q)$. The character $\chi_1$ in (\ref{dfn-add}) is called the
{\em canonical additive character} of $\gf(q)$.

A {\em multiplicative character} of $\gf(q)$ is a nonzero function
$\psi$ from $\gf(q)^*$ to the set of complex numbers such that
$\psi(xy)=\psi(x)\psi(y)$ for all pairs $(x, y) \in \gf(q)^*
\times \gf(q)^*$.
Let $g$ be a fixed primitive element of $\gf(q)$. For each
$j=0,1,\ldots,q-2$, the function $\psi_j$ with
\begin{eqnarray}\label{dfn-mul}
\psi_j(g^k)=e^{2\pi \sqrt{-1} jk/(q-1)} \ \ \mbox{for } k=0,1,\ldots,q-2
\end{eqnarray}
defines a multiplicative character of $\gf(q)$ with order $k$. When $j=0$,
$\psi_0(c)=1 \mbox{ for all }
c\in\gf(q)^*,
$
and is called the {\em trivial multiplicative
character} of $\gf(q)$.


Let $\psi$ be a multiplicative character with order $k$ where $k|(q-1)$ and $\chi$ an additive character of
$\gf(q)$. Then the {\em Gaussian sum} $G(\psi,\chi)$ of order $k$ is defined by
\begin{eqnarray}
G(\psi,\chi)=\sum_{c\in\gf(q)^*} \psi(c)\chi(c). \nonumber
\end{eqnarray}

Since $G(\psi,\chi_b)=\bar{\psi}(b)G(\psi,\chi_1)$, we just consider $G(\psi,\chi_1)$, briefly denoted as $G(\psi)$, in the sequel. If $\psi\neq\psi_0,$
 then
\begin{equation}\label{equ-Gvalue}
|G(\psi)|=q^{1/2}.
\end{equation}

\subsection{Cyclotomy}

Let $r-1=nN$ for two positive integers $n>1$ and $N>1$, and let
$\alpha$ be a fixed primitive element of $\gf(r)$.
Define $C_{i}^{(N,r)}=\alpha^i \langle \alpha^{N} \rangle$ for $i=0,1,...,N-1$, where
$\langle \alpha^{N} \rangle$ denotes the
subgroup of $\gf(r)^*$ generated by $\alpha^{N}$. The cosets $C_{i}^{(N,r)}$ are
called the {\em cyclotomic classes} of order $N$ in $\gf(r)$.
The {\em cyclotomic numbers} of order $N$ are
defined by
\begin{eqnarray*}
(i, j)_N=\left|(C_{i}^{(N,r)}+1) \cap C_{j}^{(N,r)}\right|
\end{eqnarray*}
for all $0 \leq i \leq  N-1$ and $0 \leq j \leq  N-1$.

The following lemma is proved in \cite{Stor} and will be useful in the sequel. 

\begin{lemma}\label{lem-cycnum2}
If $r \equiv 1 \pmod{4}$, we have 
$$ 
(0, 0)_2=\frac{r-5}{4}, \ (0, 1)_2=(1, 0)_2=(1, 1)_2=\frac{r-1}{4}. 
$$
If $r \equiv 3 \pmod{4}$, we have 
$$ 
(0, 1)_2=\frac{r+1}{4}, \ (0, 0)_2=(1, 0)_2=(1, 1)_2=\frac{r-3}{4}. 
$$
\end{lemma}

\subsection{Gaussian periods}

The {\em Gaussian periods} are defined by
$$
\eta_i^{(N,r)} =\sum_{x \in C_i^{(N,r)}} \chi(x), \quad i=0,1,..., N-1,
$$
where $\chi$ is the canonical additive character of $\gf(r)$.

Gaussian periods are closely related to Gaussian sums. By the discrete Fourier transform, it is known that
\begin{eqnarray}\label{equ-G sum& G period}
    \eta_i^{(N,r)} &=& \frac{1}{N}\sum\limits_{j=0}^{N-1}\zeta_N^{-ij}G(\psi^j)   \nonumber \\ 
                         &=& \frac{1}{N}\left[-1+\sum\limits_{j=1}^{N-1}\zeta_N^{-ij}G(\psi^j)\right] 
\end{eqnarray}
where $\zeta_N=e^{2\pi\sqrt{-1}/N}$ and $\psi$ is a primitive multiplicative character of order $N$ over $\gf(r)^*$.

The values of the Gaussian periods are known in a few cases, but are in general  
very hard to compute. The following is proved in \cite{DingYang}

\begin{theorem} \label{thm-wdivisibility2}
For all $i$ with $0 \le i \le N-1$, we have
$$
\left|\eta_i^{(N,r)}+\frac{1}{N}\right|\leqslant \left\lfloor{\frac{(N-1)\sqrt{r}}{N}} \right\rfloor. 
$$
\end{theorem}

\subsection{Bounds on the weights in irreducible cyclic codes}

Let $\gcd(n, q)=1$ and let $k:=\ord_n(q)$ denote the order of $q$ modulo $n$. 
Define $r=q^k$.  
Let $N>1$ be an integer dividing $r-1$, and put $n=(r-1)/N$.
Let $\alpha$ be a primitive element of $\gf(r)$ and define $\theta=\alpha^N$.
The set
\begin{eqnarray}\label{eqn-irrecode}
\C(r,N) =
\left\{\left(\tr_{r/q}(a \theta^{i}) \right)_{i=0}^{n-1}: a \in \gf(r)\right\}
\end{eqnarray}
is called an {\em irreducible cyclic $[n, k]$ code} over $\gf(q)$,
where $\tr_{r/q}$ is the trace function from $\gf(r)$ onto $\gf(q)$. 

Using Delsarte's Theorem \cite{Dels}, one can prove that the code
$\C(r,N)$ is the cyclic code with check polynomial
$m_{\theta^{-1}}(x)$, which is the minimal polynomial of $\theta^{-1}$ 
over $\gf(q)$ and is irreducible over $\gf(q)$ (see also \cite[Theorem 8.24]{LN97}).  

The determination of the weight distribution of irreducible cyclic codes 
is equivalent to that of the values of the Gaussian periods. Hence, the 
weight distribution of the irreducible cyclic codes is known for a few 
cases \cite{DingYang}, but open in general.  

The following is proved in \cite{DingYang} and will be useful in this paper. 

\begin{theorem}\label{thm-ICbounds}
Let $N$ be a positive divisor of $r-1$ and define  $N_1=\gcd((r-1)/(q-1), N)$.  Let $k$ be
the nultiplicative order of $q$ modulo $n$. Then the set
$\C(r,N)$ in (\ref{eqn-irrecode}) is a $[(q^m-1)/N, k]$ cyclic code over $\gf(q)$ in which the weight
$w$ of every nonzero codeword satisfies that
\begin{eqnarray*}
\left\lceil \frac{r - \lfloor(N_1-1)\sqrt{r}\rfloor}{qN} \right\rceil \le 
\frac{w}{q-1} \le \left\lfloor \frac{r + \lfloor(N_1-1)\sqrt{r}\rfloor}{qN}    \right\rfloor.
\end{eqnarray*}
\end{theorem}

\subsection{Lower bound on the minimum weight of a class of cyclic codes}

Let $\gcd(n, q)=1$ and let $k:=\ord_n(q)$ denote the order of $q$ modulo $n$. 
Define $r=q^k$.  
Let $N>1$ be an integer dividing $r-1$, and put $n=(r-1)/N$.
Let $\alpha$ be a primitive element of $\gf(r)$ and define $\theta=\alpha^N$.
The set
\begin{eqnarray}\label{eqn-irrecode2}
\overline{\C}(r,N) =
 \left\{ \left(\tr_{r/q}(a \theta^{i}+b)\right)_{i=0}^{n-1} :
    a, b \in \gf(r) \right\}
\end{eqnarray}
is a cyclic $[n, k+1]$ code over $\gf(q)$,
where $\tr_{r/q}$ is the trace function from $\gf(r)$ onto $\gf(q)$.  

Using Delsarte's Theorem \cite{Dels}, one can prove that the code
$\overline{\C}(r,N)$ is the cyclic code with check polynomial
$(x-1)m_{\theta^{-1}}(x)$, where $m_{\theta^{-1}}(x)$ is the minimal polynomial of $\theta^{-1}$ 
over $\gf(q)$ and is irreducible over $\gf(q)$.

\begin{theorem}\label{thm-ICbounds2}
Let $N$ be a positive divisor of $r-1$ and define  $N_1=\gcd((r-1)/(q-1), N)$.  Let $k$ be
the nultiplicative order of $q$ modulo $n$. Then the set
$\overline{\C}(r,N)$ in (\ref{eqn-irrecode2}) is a $[(q^m-1)/N, k+1, d]$ cyclic code over $\gf(q)$, where 
\begin{eqnarray*}
d \ge \min\left\{ \begin{array}{l} 
 (q-1) \left\lceil \frac{r - \lfloor(N_1-1)\sqrt{r}\rfloor}{qN} \right\rceil, \\ 
          \frac{(q-1)(r-1)-1}{qN} - \frac{q-1}{q} \left\lfloor \frac{(N-1)\sqrt{r}}{N} \right\rfloor 
\end{array}          
 \right\}. 
\end{eqnarray*}
\end{theorem}

\begin{proof} 
Let $\zeta_p=e^{2\pi \sqrt{-1}/p}$, and $\chi(x)=\zeta_p^{\tr_{r/p}(x)}$, where $\tr_{r/p}$ is
the trace function from $\gf(r)$ to $\gf(p)$. Then $\chi$ is an
additive character of $\gf(r)$. 

Let $b \in \gf(r)$. We have  
\begin{eqnarray}\label{eqn-oct41}
\sum_{y \in \gf(q)} \chi(-by)  
&=& \sum_{y \in \gf(q)} \zeta_p^{\tr_{q/p}(\tr_{r/q}(-by))} \nonumber \\
&=& \sum_{y \in \gf(q)} \zeta_p^{\tr_{q/p}(y\tr_{r/q}(-b))} \nonumber \\
&=&  \left\{ \begin{array}{ll} 
                                                          0 & \mbox{ if } \tr_{r/q}(b) \ne 0 \\
                                                          q & \mbox{ if } \tr_{r/q}(b) = 0. 
                                                         \end{array}
\right. 
\end{eqnarray}

Note that the code $\overline{\C}(r,N)$ of (\ref{eqn-irrecode2}) contains the code $\C(r,N)$ of (\ref{eqn-irrecode}) 
as a subcode. Define 
$$ 
\bc_{(a,b)} = (\tr_{r/q}(a+b), \tr_{r/q}(a \theta + b), ...,  \tr_{r/q}(a \theta^{n-1}+b)) 
$$ 
where $a, b \in \gf(r)$.  

If $\tr_{r/q}(b)=0$, $\bc_{(a,b)}$ is a codeword of the code $\C(r,N)$ of (\ref{eqn-irrecode}) 
and the Hamming weight of this codeword satisfies the bounds of  Theorem \ref{thm-ICbounds}. 
If $a=0$ and $\tr_{r/q}(b) \ne 0$, the Hamming weight of  $\bc_{(a,b)}$ is equal to $n$. 

We now consider the weight of $\bc_{(a,b)}$ for the case that $a \ne 0$ and $\tr_{r/q}(b) \ne 0$. 
Let $Z(r,a,b)$ denote the number of solutions $x \in \gf(r)$ of the equation $\tr_{r/q}(ax^{N}-b)=0$.
It then follows from (\ref{eqn-oct41}) that 
\begin{eqnarray}\label{eqn-oct42}
\lefteqn{Z(r,a,b)} \nonumber \\ 
&=& \frac{1}{q} \sum_{y \in \gf(q)} \sum_{x \in \gf(r)} \zeta_p^{\tr_{q/p}(y \tr_{r/q} (ax^{N}-b))} \nonumber \\
&=& \frac{1}{q} \sum_{y \in \gf(q)} \sum_{x \in \gf(r)} \chi(y(ax^{N}-b)) \nonumber \\
&=& \frac{1}{q} \left[ r-1 + \sum_{y \in \gf(q)^*} \sum_{x \in \gf(r)^*} \chi(yax^{N}-b) \right] \nonumber \\
&=& \frac{1}{q} \left[ r-1 + N \sum_{y \in \gf(q)^*} \sum_{x \in C_{0}^{(N,r)}} \chi(y(a x-b)) \right] 
\end{eqnarray}

Then the Hamming weight $w$ of the codeword $\bc_{(a,b)}$ is then given by 
\begin{equation}\label{eqn-wtmain}
qNw=  
(q-1)(r-1)- N \sum_{y \in \gf(q)^*} \sum_{x \in C_{0}^{(N,r)}} \chi(y(a x-b)). \nonumber 
\end{equation} 
It then follows that 
\begin{eqnarray*}
\lefteqn{w-\frac{(q-1)(r-1)-1}{qN} =} \nonumber \\
& - \frac{\sum_{y \in \gf(q)^*} \chi(-by) \left( \sum_{x \in C_{0}^{(N,r)}} \chi(ayx) + \frac{1}{N} \right)}{q}. 
\end{eqnarray*}
By Theorem \ref{thm-wdivisibility2}, we have then 
$$
\left|w-\frac{(q-1)(r-1)-1}{qN}\right| \le  \frac{q-1}{q} \left\lfloor \frac{(N-1)\sqrt{r}}{N} \right\rfloor. 
$$ 
Whence, 
\begin{eqnarray}\label{eqn-oct43}
w \ge \frac{(q-1)(r-1)-1}{qN} - \frac{q-1}{q} \left\lfloor \frac{(N-1)\sqrt{r}}{N} \right\rfloor. 
\end{eqnarray}

Combining the lower bound of (\ref{eqn-oct43}) and that of Theorem \ref{thm-ICbounds} proves 
the conclusions of this theorem.  
\end{proof}

\section{Cyclic codes from cyclotomic sequences of order four}

\subsection{Basic notations and results}\label{sec-Oct17}

Throughout this section, let $n$ be an odd prime such that $n \equiv 1 \pmod{4}$.  
It is well known that $n$ can be expressed as $n=u^2+4v^2$, where $u$ is an integer 
with $u \equiv 1 \pmod{4}$ and the sign of $v$ is undetermined. 
As usual, $q=p^m$ for a prime $p$  and satisfies $\gcd(n, q)=1$. Let $\ord_n(q)$ denote 
the multiplicative order of $q$ modulo $n$. Let $\eta$ be an $n$th primitive root 
of unity over $\gf(q^{\ord_n(q)})$. Define for each $i$ with $0 \le i \le 3$ 
$$ 
\Omega_i^{(4,n)}(x)=\prod_{i \in C_{i}^{(4,n)}} (x-\eta^i),  
$$  
where $C_{i}^{(4,n)}$ denotes the cyclotomic classes of order 4 in $\gf(n)$. 
We have 
$$ 
x^n-1=\prod_{i=0}^{3} \Omega_i^{(4,n)}(x). 
$$
It is straightforward to prove that $\Omega_i^{(4,n)}(x) \in \gf(q)[x]$ if $q \in C_{0}^{(4,n)}$.

Note that the cyclotomic classes $C_{0}^{(4,n)}$ and $C_{2}^{(4,n)}$ do not depend 
on the choice of the generator of $\gf(n)^*$ employed to define the cyclotomic classes. 
However, difference choices of the generator may lead to a swapping of $C_{1}^{(4,n)}$ 
and $C_{3}^{(4,n)}$. So we have the same conclusions for the four polynomials $\Omega_{i}^{(4,n)}(x)$. 

By definition the cyclotomic classes of order 2 are given by 
$$ 
C_{0}^{(2,n)}=C_{0}^{(4,n)} \cup C_{2}^{(4,n)}, \ C_{1}^{(2,n)}=C_{1}^{(4,n)} \cup C_{3}^{(4,n)}. 
$$
Define 
$$ 
\theta_0^{(2,n)}= \sum_{i \in C_{0}^{(2,n)}} \eta^i.  
$$
We now prove that  
\begin{eqnarray}\label{eqn-openday}
\theta_0^{(2,n)}(\theta_0^{(2,n)}+1)= \frac{n-1}{4}. 
\end{eqnarray}
In this case $-1 \in C_{0}^{(2,n)}$. By Lemma \ref{lem-cycnum2} we have 
\begin{eqnarray*} 
\left(\theta_0^{(2,n)}\right)^2 
&=& \left(\sum_{i \in C_0^{(2,n)}} \eta^i \right) \left(\sum_{j \in C_0^{(2,n)}} \eta^j \right) \\
&=& 
     \left( \sum_{i \in C_0^{(2,n)}} \eta^i \right)\left( \sum_{j \in C_0^{(2,n)}} \eta^{-j} \right) \\  
&=& \frac{n-1}{2} + 
      \sum_{i,j \in C_0^{(2,n)} \atop i \ne j} \eta^{i-j}  \\       
&=& \frac{n-1}{4} - \theta_0^{(2,n)}. 
\end{eqnarray*} 
Hence,  $\theta_0^{(2,n)} \in \{0, -1\}$ if and only if $(n-1)/4 \equiv 0 \pmod{p}$.

The following lemma will be useful in this section and 
can be proved with the Law of Biquadratic Reciprocity. 

\begin{lemma}\label{lem-biqreciprocity} 
We have the following conclusions: 
\begin{itemize}
\item 2 is a biquadratic residue modulo $n \equiv 1 \pmod{4}$ if and only if $n=a^2+64b^2$ for some integers $a$ and $b$. 
\item 3 is a biquadratic residue modulo $n \equiv 1 \pmod{4}$ if and only if $n=a^2+4b^2$ for some integers $a$ and $b$ and either  
\begin{enumerate} 
  \item $n \equiv 1 \pmod{8}$ and $b \equiv 0 \pmod{3}$, or 
  \item $n \equiv 5 \pmod{8}$ and $a \equiv 0 \pmod{3}$.   
\end{enumerate}  
\item 5 is a biquadratic residue modulo $n=a^2+b^2$, where $b$ is even, if and only if $b \equiv 0 \pmod{5}$. 
\end{itemize}
\end{lemma}

\begin{lemma}\label{lem-O4bound}
Let $\ord_n(q)=(n-1)/4$ and $q-1<n$. Assume that $q \in C_0^{(4,n)}$. Then the cyclic 
code over $\gf(q)$ with parity check polynomial $\Omega_i^{(4,n)}(x)$ has parameters  
$[n, (n-1)/4, d_i]$, where   
\begin{eqnarray*}
d \ge (q-1) \left\lceil \frac{q^{\frac{n-1}{4}} - \left\lfloor(N_1-1)\sqrt{q^{\frac{n-1}{4}}}\right\rfloor}{qN} \right\rceil   
\end{eqnarray*}
and  
$$ 
N=\frac{q^{\frac{n-1}{4}} -1}{n} \mbox{ and } N_1=\frac{N}{q-1}.  
$$
\end{lemma} 

\begin{proof} 
Since $\ord_n(q)=(n-1)/4$ and $q \in C_0^{(4,n)}$, the four polynomials $\Omega_i^{(4,n)}(x)$ 
are irreducible and over $\gf(q)$. Hence the code with parity check polynomial 
$\Omega_i^{(4,n)}(x)$ is an irreducible cyclic code with dimension $(n-1)/4$. 

Note that $q-1<n$ and $n$ is prime. We have then   
$$ 
\gcd\left(\frac{q^{\frac{n-1}{4}}-1}{q-1}, N \right)=\frac{N}{q-1}. 
$$
The desired bounds on the nonzero weights follow from Theorem \ref{thm-ICbounds}.  
\end{proof} 

\begin{example} 
Let $q=3$ and $n=13$. We have then the canonical factorization 
\begin{eqnarray*} 
x^{13}-1 &=& (x+2)(x^3+2x+2)(x^3+2x^2+2) \times \\
               &  & (x^3 + x^2 + x + 2)(x^3 + 2x^2 + 2x + 2).  
\end{eqnarray*}
The cyclic code with parity check polynomial $x^3 + 2x + 2$ has parameters $[13, 3, 9]$. 

In this case $N=2$ and $N_1=1$. The lower and upper bound in Lemma \ref{lem-O4bound} are 
equal to $9$.  

In general, the bounds are tight if $N_1$ is small.   
\end{example}

\begin{lemma}\label{lem-O4bound2}
Let $\ord_n(q)=(n-1)/4$ and $q-1<n$. Assume that $q \in C_0^{(4,n)}$. Then the cyclic 
code over $\gf(q)$ with parity check polynomial $(x-1)\Omega_i^{(4,n)}(x)$ has parameters  
$[n, (n+3)/4, d_i]$, where  
\begin{eqnarray*}
d_i \ge  \frac{(q-1)(q^{\frac{n-1}{4}}-1)-1}{qN} - \frac{q-1}{q} \left\lfloor \frac{(N-1)\sqrt{q^{\frac{n-1}{4}} }}{N} \right\rfloor 
\end{eqnarray*}
and  
$$ 
N=\frac{q^{\frac{n-1}{4}} -1}{n}.  
$$
\end{lemma} 

\begin{proof} 
Since $\ord_n(q)=(n-1)/4$ and $q \in C_0^{(4,n)}$, the four polynomials $\Omega_i^{(4,n)}(x)$ 
are irreducible and over $\gf(q)$. Hence the code with parity check polynomial 
$(x-1)\Omega_i^{(4,n)}(x)$ has dimension $(n+3)/4$, and is the same as the code of 
 (\ref{eqn-irrecode2}).   

Note that $q-1<n$ and $n$ is prime. We have then   
$$ 
\gcd\left(\frac{q^{\frac{n-1}{4}}-1}{q-1}, N \right)=\frac{N}{q-1}. 
$$
The desired lower bound on the minimum weight follow from Theorem \ref{thm-ICbounds2}.  
\end{proof} 

\begin{example} 
Let $q=3$ and $n=13$. We have then the canonical factorization 
\begin{eqnarray*} 
x^{13}-1 &=& (x+2)(x^3+2x+2)(x^3+2x^2+2) \times \\
               &  & (x^3 + x^2 + x + 2)(x^3 + 2x^2 + 2x + 2).  
\end{eqnarray*}
The cyclic code with parity check polynomial $(x^3 + 2x + 2)(x-1)$ has parameters $[13, 4, 7]$. 

In this case $N=2$ and $N_1=1$. The lower and upper bound in Lemma \ref{lem-O4bound} are  
equal to $7$.  

In general, the lower bound is tight if $N$ is small.   
\end{example}

\subsection{The first class of cyclic codes from cyclotomic sequences of order 4}\label{sec-codecyc4}

Define   
\begin{eqnarray}\label{eqn-seqorder4}  
\lambda_i=\left\{ \begin{array}{ll}  
                            1 & \mbox{ if } i \bmod{n} \in C_{0}^{(4,n)} \cup C_{1}^{(4,n)}  \\ 
                            0 & \mbox{ otherwise} 
                           \end{array} 
               \right. 
\end{eqnarray} 
for all $i \ge 0$. This $\lambda^{\infty}$ was defined as a binary sequence in \cite{DHL99} 
and was proved to have optimal autocorrelation under certain condition. Here in this section, 
we treat it as a sequence over $\gf(q)$ for any prime power $q$, and employ it to construct 
cyclic codes. 

We define 
\begin{eqnarray*} 
\Lambda(x) &=& \sum_{i \in C_{0}^{(4,n)} \cup C_{1}^{(4,n)}} x^i \in \gf(q)[x], \\
\Gamma(x) &=& \sum_{i \in C_{1}^{(4,n)} \cup C_{2}^{(4,n)}} x^i \in \gf(q)[x].  
\end{eqnarray*} 
Both $\Lambda(x)$ and $\Gamma(x)$ depend on the choice of the generator of $\gf(n)^*$ 
employed to define the cyclotomic classes of order 4. 

Notice that 
$$ 
\left( \sum_{i \in C_{0}^{(4,n)}}  +  \sum_{i \in C_{1}^{(4,n)}} + 
       \sum_{i \in C_{2}^{(4,n)}}  +  \sum_{i \in C_{3}^{(4,n)}} \right) \eta^i=-1. 
$$
We have then 
\begin{eqnarray}\label{eqn-opendayall}
\Lambda(\eta^i)=\left\{ \begin{array}{ll} 
\Lambda(\eta)  & \mbox{ if } i \in C_{0}^{(4,n)} \\  
\Gamma(\eta)  & \mbox{ if } i \in C_{1}^{(4,n)} \\  
-(\Lambda(\eta)+1)  & \mbox{ if } i \in C_{2}^{(4,n)} \\  
-(\Gamma(\eta)+1)  & \mbox{ if } i \in C_{3}^{(4,n)}.   
                               \end{array} 
                    \right.                                          
\end{eqnarray} 

We have also that 
\begin{eqnarray}\label{eqn-opendayall2}
\Lambda(\eta^0)=\Lambda(1)=\frac{n-1}{2} \bmod{p}.            
\end{eqnarray}

\begin{theorem}\label{thm-bilegendre}
Let $\frac{n-1}{4} \equiv 0 \pmod{p}$, and 
let $\lambda^{\infty}$ be the sequence of period $n$ over $\gf(q)$ defined in (\ref{eqn-seqorder4}). 
As before, $n=u^2+4v^2$ with $u \equiv 1 \pmod{4}$.  
\begin{enumerate}
\item $n \equiv 1 \pmod{8}$. 

When $\frac{v}{2} \not\equiv 0 \pmod{p}$, we have 
$ 
\ls_\lambda=n-1   
$
and 
$$ 
m_\lambda(x)=\frac{x^n-1}{x-1}. 
$$
In this subcase, the cyclic code $\calC_\lambda$  over $\gf(q)$ defined by the sequence $\lambda^{\infty}$ 
has the generator polynomial $m_\lambda(x)$ above and parameters $\left[n, 1, n \right]$.  

When $\frac{v}{2} \equiv 0 \pmod{p}$ and $q \in C_{0}^{(4,n)}$, we have $\ls_\lambda=(n-1)/2$  
and 
\begin{eqnarray*}
m_\lambda(x)=\left\{ \begin{array}{ll} 
\Omega_2^{(4,n)}(x) \Omega_3^{(4,n)}(x) & \mbox{ if } \left\{\begin{array}{ll} 
                                                                                               \Lambda(\eta)=0 \\ 
                                                                                               \Gamma(\eta)=0 
                                                                                               \end{array} \right. \\  
\Omega_1^{(4,n)}(x) \Omega_2^{(4,n)}(x) &  \mbox{ if } \left\{\begin{array}{ll} 
                                                                                               \Lambda(\eta)=0 \\ 
                                                                                               \Gamma(\eta)=-1 
                                                                                               \end{array} \right. \\  
\Omega_0^{(4,n)}(x) \Omega_3^{(4,n)}(x) &  \mbox{ if } \left\{\begin{array}{ll} 
                                                                                               \Lambda(\eta)=-1 \\ 
                                                                                               \Gamma(\eta)=0 
                                                                                               \end{array} \right. \\  
\Omega_0^{(4,n)}(x) \Omega_1^{(4,n)}(x) &  \mbox{ if } \left\{\begin{array}{ll} 
                                                                                               \Lambda(\eta)=-1 \\ 
                                                                                               \Gamma(\eta)=-1. 
                                                                                               \end{array} \right.   
                               \end{array} 
                    \right.                                          
\end{eqnarray*} 
In this subcase, the cyclic code $\calC_\lambda$  over $\gf(q)$ defined by the sequence $\lambda^{\infty}$ 
has the generator polynomial $m_\lambda(x)$ above and parameters $\left[n, (n+1)/2, d \right]$. 
In addition, the minimum odd-like weight $d_{odd} \ge \sqrt{n}$.

\item $n \equiv 5 \pmod{8}$. 

When $\frac{u^2+3}{4} \not\equiv 0 \pmod{p}$, we have 
$ 
\ls_\lambda=n-1   
$
and 
$$ 
m_\lambda(x)=\frac{x^n-1}{x-1}. 
$$
In this subcase, the cyclic code $\calC_\lambda$  over $\gf(q)$ defined by the sequence $\lambda^{\infty}$ 
has the generator polynomial $m_\lambda(x)$ above and parameters $\left[n, 1, n \right]$.

When $\frac{u^2+3}{4} \equiv 0 \pmod{p}$ and $q \in C_{0}^{(4,n)}$, we have $\ls_\lambda=3(n-1)/4$  
and 
\begin{eqnarray*}
m_\lambda(x)=\left\{ \begin{array}{ll} 
\frac{x^n-1}{(x-1) \Omega_0^{(4,n)}(x)} & \mbox{ if } \Lambda(\eta)=0  \\  
\frac{x^n-1}{(x-1) \Omega_2^{(4,n)}(x)} & \mbox{ if } \Lambda(\eta)=-1  \\  
\frac{x^n-1}{(x-1) \Omega_1^{(4,n)}(x)} & \mbox{ if } \Gamma(\eta)=0  \\  
\frac{x^n-1}{(x-1) \Omega_3^{(4,n)}(x)} & \mbox{ if } \Gamma(\eta)=-1.   \\ 
                               \end{array} 
                    \right.                                          
\end{eqnarray*} 
In this subcase, the cyclic code $\calC_\lambda$  over $\gf(q)$ defined by the sequence $\lambda^{\infty}$ 
has the generator polynomial $m_\lambda(x)$ above and parameters $\left[n, (n+3)/4, d \right]$.    
Furthermore, the minimum weight $d$ has the lower bound of Lemma \ref{lem-O4bound2} if $\ord_{n}(q)=(n-1)/4$. 

\end{enumerate}
\end{theorem}

\begin{proof} 
To prove this theorem, we need information on cyclotomic numbers of order 4. 
When $n \equiv 5 \pmod{8}$ is odd, the relation between the 16 cyclotomic 
numbers of order 4 is given 
by the following Table \ref{tab-adset1} \cite{Stor}:   
\begin{table}[ht]
\begin{center}
\begin{tabular}{|c|c|c|c|c|} \hline 
$(h,k)_4$      &  0  &  1  &  2  &  3    \\ \hline 
0            &  A  &  B  &  C  &  D \\ \hline 
1            &  E  &  E  &  D  &  B \\ \hline 
2            &  A  &  E  &  A  &  E \\ \hline 
3            &  E  &  D  &  B  &  E \\ \hline  
\end{tabular}
\caption{The relations of the cyclotomic numbers of order 4, when $n \equiv 5 \pmod{8}$.
}\label{tab-adset1}
\end{center}
\end{table} 

Thus, there are five possible different cyclotomic numbers in this case; i.e., 
\begin{eqnarray*}
& & A=\frac{n-7+2u}{16},\\
& & B=\frac{n+1+2u-8v}{16},\\
& & C=\frac{n+1-6u}{16},\\
& & D=\frac{n+1+2u+8v}{16},\\
& & E=\frac{n-3-2u}{16}.  
\end{eqnarray*} 

When $n \equiv 1 \pmod{8}$, the relation between the 16 cyclotomic numbers is given 
by the following Table \ref{tab-adset2} \cite{Stor}:   

\begin{table}[ht]
\begin{center}
\begin{tabular}{|c|c|c|c|c|} \hline 
$(h,k)_4$      &  0  &  1  &  2  &  3    \\ \hline 
0            &  A  &  B  &  C  &  D \\ \hline 
1            &  B  &  D  &  E  &  E \\ \hline 
2            &  C  &  E  &  C  &  E \\ \hline 
3            &  D  &  E  &  E  &  B \\ \hline  
\end{tabular}
\caption{The relations of the cyclotomic numbers of order 4, when $n \equiv 1 \pmod{8}$.
}\label{tab-adset2}
\end{center}
\end{table} 

Thus, there are five possible different cyclotomic numbers in this case; i.e., 
\begin{eqnarray*}
& & A=\frac{n-11-6u}{16},\\
& & B=\frac{n-3+2u+8v}{16},\\
& & C=\frac{n-3+2u}{16},\\
& & D=\frac{n-3+2u-8v}{16},\\
& & E=\frac{n+1-2u}{16}.  
\end{eqnarray*}

To determine the minimal polynomial $m_\lambda(x)$, we need to compute 
$\gcd(\Lambda(x), x^n-1)$. 

We first prove the conclusions for the case that $n \equiv 1 \pmod{8}$. In this case  
$-1 \in C_0^{(4,n)}$ and $v$ must be even. Note that $\frac{n-1}{4} \equiv 0 \pmod{p}$. It then follows from the 
relations of the cyclotomic numbers and the cyclotomic numbers above that  
\begin{eqnarray*} 
\lefteqn{\Lambda(\eta)^2 } \\ 
&=& \sum_{i \in C_{0}^{(4,n)}, j \in C_{0}^{(4,n)}}  \eta^{i-j} + 
          \sum_{i \in C_{0}^{(4,n)}, j \in C_{1}^{(4,n)}}  \eta^{i-j} + \\ 
& &    \sum_{i \in C_{1}^{(4,n)}, j \in C_{0}^{(4,n)}}  \eta^{i-j} + 
          \sum_{i \in C_{1}^{(4,n)}, j \in C_{1}^{(4,n)}}  \eta^{i-j} \\     
&=& ((0,0)_4 + (1, 0)_4 + (0, 1)_4 + (1,1)_4) \sum_{i\in C_{0}^{(4,n)}}\eta^i   + \\
& & ((3,3)_4 + (0, 3)_4 + (3, 0)_4 + (0,0)_4) \sum_{i\in C_{1}^{(4,n)}}\eta^i   + \\
& & ((2,2)_4 + (3, 2)_4 + (2, 3)_4 + (3,3)_4) \sum_{i\in C_{2}^{(4,n)}}\eta^i   + \\ 
& & ((1,1)_4 + (2, 1)_4 + (1, 2)_4 + (2,2)_4) \sum_{i\in C_{3}^{(4,n)}}\eta^i   + \\ 
& & \frac{n-1}{2}  \\ 
&=& (A+2B+D) \sum_{i\in C_{0}^{(4,n)}}\eta^i   + 
    (A+B+2D) \sum_{i\in C_{1}^{(4,n)}}\eta^i   + \\
& & (B+C+2E) \sum_{i\in C_{2}^{(4,n)}}\eta^i   +  
    (C+D+2E) \sum_{i\in C_{3}^{(4,n)}}\eta^i    \\
& &  + \frac{n-1}{2}   \\    
&=& \frac{n-5+2v}{4} \sum_{i\in C_{0}^{(4,n)}}\eta^i   + 
    \frac{n-5-2v}{4}  \sum_{i\in C_{1}^{(4,n)}}\eta^i   + \\
& & \frac{n-1+2v}{4} \sum_{i\in C_{2}^{(4,n)}}\eta^i   +  
    \frac{n-1-2v}{4} \sum_{i\in C_{3}^{(4,n)}}\eta^i    \\ 
& & + \frac{n-1}{2}   \\  
&=& -\Lambda(\eta)+ \frac{n-1}{4} + \frac{v}{2} \left(2 \sum_{i\in C_{0}^{(2,n)}}\eta^i + 1\right) \\
&=& -\Lambda(\eta) + \frac{v}{2} \left(2 \sum_{i\in C_{0}^{(2,n)}}\eta^i + 1\right) \\               
\end{eqnarray*} 
Whence, 
\begin{eqnarray}\label{eqn-openday1} 
\Lambda(\eta) (\Lambda(\eta) +1) = \frac{v}{2} \left(2 \sum_{i\in C_{0}^{(2,n)}}\eta^i  + 1\right).                 
\end{eqnarray} 

Note that $\frac{n-1}{4} \equiv 0 \pmod{p}$.  
By (\ref{eqn-openday}) we have  
$ 
\sum_{i\in C_{0}^{(2,n)}}\eta^i  \in \{0, -1\}. 
$
It then follows that 
\begin{eqnarray}\label{eqn-openday2} 
2\sum_{i\in C_{0}^{(2,n)}}\eta^i + 1 \in \{1, -1\}. 
\end{eqnarray}

Similarly, one can show that 
\begin{eqnarray}\label{eqn-openday3} 
\Gamma(\eta) (\Gamma(\eta) +1) = -\frac{v}{2} \left(2 \sum_{i\in C_{0}^{(2,n)}}\eta^i  + 1\right).                 
\end{eqnarray} 

The desired conclusions on the linear span and the minimal polynomial of the sequence 
$\lambda^\infty$ for Case 1 then follow from (\ref{eqn-opendayall}), (\ref{eqn-opendayall2}), 
(\ref{eqn-openday1}), (\ref{eqn-openday2}), (\ref{eqn-openday3}), and Lemma \ref{lem-ls0}.  
The dimension and the generator polynomial of the code $\C_\lambda$ follow from the 
conclusions on the linear span and the minimal polynomial of the sequence and the definition 
of the code $\C_\lambda$.  
In the first subcase, it is obvious that the minimum nonzero weight $d=n$.  
In the second subcase, the generator polynomial of the code shows that 
$\calC_\lambda$ is a duadic code. So we have the square-root bound on 
the minimum odd-like weight \cite{HPbook,Leon84,Ples87,DLX,DP}.

We now prove the conclusions for Case 2. Since $n \equiv 5 \pmod{8}$, $-1 \in C_2^{(4,n)}$. 
In this case $v$ must be odd. 
Note that $\frac{n-1}{4} \equiv 0 \pmod{p}$. It then follows from the 
relations of the cyclotomic numbers and the cyclotomic numbers above that  
\begin{eqnarray*} 
\lefteqn{\Lambda(\eta)^2} \\  
&=& \sum_{i \in C_{0}^{(4,n)}, j \in C_{2}^{(4,n)}}  \eta^{i-j} + 
          \sum_{i \in C_{0}^{(4,n)}, j \in C_{3}^{(4,n)}} \eta^{i-j} + \\
& &          \sum_{i \in C_{1}^{(4,n)}, j \in C_{2}^{(4,n)}} \eta^{i-j} + 
          \sum_{i \in C_{1}^{(4,n)}, j \in C_{3}^{(4,n)}} \eta^{i-j}   \\ 
&=& ((2,0)_4 + (3, 0)_4 + (2, 1)_4 + (3,1)_4) \sum_{i\in C_{0}^{(4,n)}}\eta^i   + \\
& & ((1,3)_4 + (2, 3)_4 + (1, 0)_4 + (2,0)_4) \sum_{i\in C_{1}^{(4,n)}}\eta^i   + \\
& & ((0,2)_4 + (1, 2)_4 + (0, 3)_4 + (1,3)_4) \sum_{i\in C_{2}^{(4,n)}}\eta^i   + \\ 
& & ((3,1)_4 + (0, 1)_4 + (3, 2)_4 + (0,2)_4) \sum_{i\in C_{3}^{(4,n)}}\eta^i    \\ 
&=& (A+D+2E) \sum_{i\in C_{0}^{(4,n)}}\eta^i   + \\
& &    (A+B+2E) \sum_{i\in C_{1}^{(4,n)}}\eta^i   + \\
& & (B+C+2D) \sum_{i\in C_{2}^{(4,n)}}\eta^i   +  \\
& &    (C+D+2B) \sum_{i\in C_{3}^{(4,n)}}\eta^i   + \frac{n-1}{2}   \\    
&=& -\Lambda(\eta) -\frac{n-1}{4} + \frac{v \left(2 \sum_{i\in C_{0}^{(2,n)}}\eta^i + 1\right) -1}{2}\\               
&=& -\Lambda(\eta)+ \frac{v \left(2 \sum_{i\in C_{0}^{(2,n)}}\eta^i + 1\right) -1}{2}.                
\end{eqnarray*} 
Whence, 
\begin{eqnarray}\label{eqn-openday12} 
\Lambda(\eta) (\Lambda(\eta) +1) = \frac{v \left(2 \sum_{i\in C_{0}^{(2,n)}}\eta^i + 1\right) -1}{2}.  
\end{eqnarray} 

Similarly, one can show that 
\begin{eqnarray}\label{eqn-openday32} 
\Gamma(\eta) (\Gamma(\eta) +1) = -\frac{v \left(2 \sum_{i\in C_{0}^{(2,n)}}\eta^i + 1\right) +1}{2}.     
\end{eqnarray} 

Since $n \equiv 5 \pmod{8}$ and $p$ divides $(n-1)/4$, $p$ must be odd. Note that 
$$ 
\frac{n-1}{4}=\frac{u^2+3}{4} + (|v|-1)(|v|+1). 
$$
Hence, $\frac{u^2+3}{4} \equiv 0 \pmod{p}$ if and only if $(|v|-1)(|v|+1) \equiv 0 \pmod{p}$.  
However, $(|v|-1)(|v|+1) \equiv 0 \pmod{p}$ if and only if $p$ divides one and only one of 
$|v|-1$ and $|v|+1$.  

The desired conclusions on the linear span and the minimal polynomial of the sequence $\lambda^{\infty}$ 
for Case 2 then follow from (\ref{eqn-opendayall}), (\ref{eqn-opendayall2}), 
(\ref{eqn-openday12}), (\ref{eqn-openday2}), (\ref{eqn-openday32}), and Lemma \ref{lem-ls0}.  
The dimension and the generator polynomial of the code $\C_\lambda$ follow from the 
conclusions on the linear span and the minimal polynomial of the sequence and the definition 
of the code $\C_\lambda$.  
In the first subcase, it is obvious that the minimum nonzero weight $d=n$.  
In the second subcase, the format of the generator polynomial of the code 
shows that the minimum weight $d$ has the lower bound of Lemma \ref{lem-O4bound2} 
if $\ord_{n}(q)=(n-1)/4$. 
\end{proof}

\begin{example} 
Let $(p, m, n)=(2,1,73)$.  Then $q=2 \in C_0^{(4,n)}$ and $n=u^2+4v^2=(-3)^2+4\times 4^2$. 
Hence $v/2 \bmod{p} =0$. 
Then $\calC_\lambda$ is a $[73,37,12]$ cyclic code over $\gf(q)$ 
with generator polynomial 
\begin{eqnarray*} 
x^{36} + x^{35} + x^{34} + x^{32} + x^{31} + x^{29} + x^{28} + \\ 
x^{27} + x^{25} + x^{23} + x^{18} + x^{13} + x^{11} + x^9 +\\  
x^8 + x^7 + x^5 + x^4 + x^2 + x + 1. 
\end{eqnarray*} 
The best binary linear code known of length 73 and dimension 37 has minimum weight 14.  
\end{example} 

\begin{example} 
Let $(p, m, n)=(2,1,89)$.  Then $q=2 \in C_0^{(4,n)}$ and $n=u^2+4v^2=5^2+4\times 4^2$. 
Hence $v/2 \bmod{p} =0$. 
Then $\calC_\lambda$ is a $[89,45,15]$ cyclic code over $\gf(q)$ 
with generator polynomial 
\begin{eqnarray*} 
x^{44} + x^{43} + x^{42} + x^{41} + x^{40} + x^{35} + x^{34} + \\ 
x^{33} + x^{31} + x^{26} +   x^{24} + x^{23} + x^{22} + x^{21} + \\ 
x^{20} + x^{18} + x^{13} + x^{11} + x^{10} + x^9 + x^4 + \\
    x^3 + x^2 + x + 1.  
\end{eqnarray*} 
The best binary linear code known of length 89 and dimension 45 has minimum weight 17.  
\end{example}

\begin{example}\label{exam-order4balance1}  
Let $(p, m, n)=(3,1,13)$.  Then $q=3 \in C_0^{(4,n)}$ and $n=u^2+4v^2=(-3)^2+4\times 1^2$. 
Hence $(u^2+3)/4 \bmod{p} =0$. 
Then $\calC_\lambda$ is a $[13,4,7]$ cyclic code over $\gf(q)$ 
with generator polynomial 
\begin{eqnarray*} 
x^9 + x^7 + x^6 + 2x^4 + x^2 + 2x + 2. 
\end{eqnarray*} 
This code is optimal. 
\end{example} 

\begin{example}\label{exam-order4balance2}  
Let $(p, m, n)=(7,1,29)$.  Then $q=7 \in C_0^{(4,n)}$ and $n=u^2+4v^2=5^2+4\times 1^2$. 
Hence $(u^2+3)/4 \bmod{p} =0$.   
Then $\calC_\lambda$ is a $[29,8,15]$ cyclic code over $\gf(q)$ 
with generator polynomial 
\begin{eqnarray*} 
x^{21} + 2x^{20} + 2x^{19} + 6x^{18} + x^{17} + 4x^{16} + 4x^{15} + \\ 
4x^{13} +  2x^{12} +  6x^{11} + 5x^{10} + x^9 + 2x^8 + 3x^7 + \\ 
3x^6 + x^5 + 4x^3 + 2x^2   + x + 6. 
\end{eqnarray*} 
This is the best cyclic code over $\gf(q)$ with length 29 and dimension 8. The best linear code 
over $\gf(q)$ with length 29 and dimension 8 has minimum weight 17. 
\end{example} 

\begin{remark} 
It was proved in \cite{DHL99} that the sequence $\lambda^{\infty}$ defined in (\ref{eqn-seqorder4}) 
has optimal autocorrelation and the set $C_0^{(4,n)} \cup C_1^{(4,n)}$ is an $(n, (n-1)/2, (n-5)/4, (n-1)/2)$ 
almost difference set in $\gf(n)$  when $v=\pm 1$. Examples \ref{exam-order4balance1} and 
\ref{exam-order4balance2} demonstrate that the cyclic codes defined by the almost difference sets 
have good parameters.   
\end{remark} 

\begin{open} 
Determine the parameters of the code $\C_\lambda$ defined by the sequence $\lambda^{\infty}$ of 
(\ref{eqn-seqorder4}) for the case that $\frac{n-1}{4} \not\equiv 0 \pmod{p}$. 
\end{open}

\subsection{The second class of cyclic codes from cyclotomic sequences of order 4}\label{sec-codecyc42} 

Unless otherwise stated, the symbols and notations of this section are the same as those in Section 
\ref{sec-codecyc4}. In this section, we always assume that $q \in C_{0}^{(4,n)}$. This ensures that 
the polynomials $\Omega_i^{(4,n)}(x)$ defined in Section \ref{sec-Oct17} are over $\gf(q)$. In 
this section, we also assume that $\frac{n-1}{4} \bmod{p}=0$. 
Our task of this section is to construct more cyclic codes over $\gf(q)$ using two cyclotomic 
sequences of order four.  

The two sequences we will employ in this section are defined by 
\begin{eqnarray}\label{eqn-seqorder42}  
\lambda_i=\left\{ \begin{array}{ll}  
                            1 & \mbox{ if } i \bmod{n} \in C_{1}^{(4,n)} \cup C_{2}^{(4,n)}  \cup C_{3}^{(4,n)}  \\ 
                            0 & \mbox{ if } i \bmod{n} \in C_{0}^{(4,n)} \\ 
                            \rho & \mbox{ if }  i \bmod{n}=0 
                           \end{array} 
               \right. 
\end{eqnarray} 
for all $i \ge 0$, where $\rho \in \{0, 1\}$. These two sequences $\lambda^{\infty}$ are characterized 
by the cyclotomic class $C_{0}^{(4,n)}$, and are viewed as sequences over $\gf(q)$ for any prime power 
$q$.

We define 
\begin{eqnarray*} 
\Lambda(x) = \rho + \sum_{i \in C_{1}^{(4,n)} \cup C_{2}^{(4,n)} \cup C_{3}^{(4,n)}} x^i \in \gf(q)[x]. 
\end{eqnarray*} 

Let $\eta$ be an $n$th primitive root of unity over $\gf(q^{\ord_n(q)})$. We define 
$$
\eta_i = \sum_{\ell \in C_i^{(4,n)}} \eta^\ell 
$$
for each $i\in \{0, 1, 2, 3\}$. Because of the assumption 
that $\frac{n-1}{4} \bmod{p} =0$,  by  (\ref{eqn-openday}) we have 
\begin{eqnarray}\label{eqn-sept281}
\eta_0+\eta_2=\sum_{i \in C_0^{(4, n)} \cup C_2^{(4, n)} } \eta^i = \sum_{i \in C_0^{(2, n)} } \eta^i  \in \{0, -1\}.
\end{eqnarray}
The value $\eta_0+\eta_2$ depends on the choice of $\eta$. 
Throughout this section, we fix an $\eta$ such that $\eta_0+\eta_2=0$. 
Notice that 
$$ 
\eta_0+\eta_1+\eta_2+\eta_3=-1. 
$$
We have then 
$$
\eta_1+\eta_3=-1.
$$ 
It is easily seen that   
$$ 
\Lambda(\eta)=\rho-1-\sum_{i \in C_0^{(4,n)}} \eta^i 
$$
and 
\begin{eqnarray}\label{eqn-opendayall22b}
\Lambda(\eta^j)=\rho -1 - \eta_i 
\end{eqnarray} 
if $j \in C_i^{(4,n)}$.

Due to the assumption that $\frac{n-1}{4} \bmod{p} =0$,  
\begin{eqnarray}\label{eqn-opendayall22}
\Lambda(\eta^0)=\Lambda(1)=\rho.            
\end{eqnarray}

When $n \equiv 1 \pmod{8}$, the linear span and minimal polynomial of the sequence $\lambda^{\infty}$ 
as well as the parameters of the code $\C_\lambda$ are given in the following theorem.

\begin{theorem}\label{thm-bilegendre81}
Let $\frac{n-1}{4} \equiv 0 \pmod{p}$ and $q \in C_0^{(4,n)}$, and let $n \equiv 1 \pmod{8}$. 
let $\lambda^{\infty}$ be the sequence of period $n$ over $\gf(q)$ defined in (\ref{eqn-seqorder42}). 
As before, $n=u^2+4v^2$ with $u \equiv 1 \pmod{4}$.  
\begin{enumerate}
\item When $\frac{n+1-2u}{16} \equiv 0 \pmod{p}$ and $\frac{n-3+2u}{16} \equiv 0 \pmod{p}$, 
\begin{eqnarray*}
m_\lambda(x)=\left\{ \begin{array}{ll} 
\frac{x^n-1}{(x-1)\Omega_3^{(4,n)}(x)} &  \mbox{ if } \left\{\begin{array}{ll} 
                                                                                                \eta_1=0 \\ 
                                                                                                 \rho=0 
                                                                                               \end{array} \right. \\   
\frac{x^n-1}{(x-1)\Omega_1^{(4,n)}(x)} & \mbox{ if } \left\{\begin{array}{ll} 
                                                                                                \eta_1=-1 \\ 
                                                                                                 \rho=0 
                                                                                               \end{array} \right. \\   
\frac{x^n-1}{\Omega_1^{(4,n)}(x) \Omega_0^{(4,n)}(x) \Omega_2^{(4,n)}(x)} & \mbox{ if } \left\{\begin{array}{ll} 
                                                                                                \eta_1=0 \\ 
                                                                                                 \rho=1 
                                                                                               \end{array} \right. \\   
\frac{x^n-1}{\Omega_3^{(4,n)}(x) \Omega_0^{(4,n)}(x) \Omega_2^{(4,n)}(x)} & \mbox{ if } \left\{\begin{array}{ll} 
                                                                                                \eta_1=-1 \\ 
                                                                                                 \rho=1.  
                                                                                               \end{array} \right.   
                               \end{array} 
                    \right.                                          
\end{eqnarray*} 
and   
\begin{eqnarray*}
\ls_\lambda=\left\{ \begin{array}{ll} 
n-\frac{n+3}{4}  & \mbox{ if } \eta_1=0 \mbox{ and } \rho=0 \\  
n-\frac{n+3}{4} & \mbox{ if } \eta_1=-1 \mbox{ and } \rho=0 \\  
n- \frac{3n-3}{4} & \mbox{ if } \eta_1=0 \mbox{ and } \rho=1 \\  
n- \frac{3n-3}{4}  & \mbox{ if } \eta_1=-1 \mbox{ and } \rho=1. \\  
                               \end{array} 
                    \right.                                          
\end{eqnarray*} 
In this case, the cyclic code $\calC_\lambda$  over $\gf(q)$ defined by the sequence $\lambda^{\infty}$ has 
the generator polynomial $m_\lambda(x)$ and parameters $\left[n, n-\ls_\lambda, d \right]$. 
In addition, if $\eta_1=0$ and  $\rho=0$ or $\eta_1=-1$ and  $\rho=0$, 
the minimum weight $d$ of the code has the lower bound of Lemma \ref{lem-O4bound2}, 
provided that $\ord_{n}(q)=(n-1)/4$.

\item  When $\frac{n+1-2u}{16} \equiv 0 \pmod{p}$ and $\frac{n-3+2u}{16} \not\equiv 0 \pmod{p}$, 
\begin{eqnarray*}
m_\lambda(x)=\left\{ \begin{array}{ll} 
\frac{x^n-1}{x-1} & \mbox{ if }  \rho=0 \\  
\frac{x^n-1}{\Omega_0^{(4,n)}(x) \Omega_2^{(4,n)}(x)} & \mbox{ if }  \rho=1  
                               \end{array} 
                    \right.                                          
\end{eqnarray*} 
and   
\begin{eqnarray*}
\ls_\lambda=\left\{ \begin{array}{ll} 
n-1  & \mbox{ if }  \rho=0 \\  
n- \frac{n-1}{2}  & \mbox{ if }  \rho=1. \\  
                               \end{array} 
                    \right.                                          
\end{eqnarray*} 
In this case, the cyclic code $\calC_\lambda$  over $\gf(q)$ defined by the sequence $\lambda^{\infty}$ has 
the generator polynomial $m_\lambda(x)$ and parameters $\left[n, n-\ls_\lambda, d \right]$, where
\begin{eqnarray*}
\left\{ \begin{array}{ll} 
d=n  & \mbox{ if }  \rho=0 \\  
d \ge \sqrt{n}   & \mbox{ if }  \rho=1. \\  
                               \end{array} 
                    \right.                                          
\end{eqnarray*}

\item When $\frac{n+1-2u}{16} \equiv 1 \pmod{p}$ and $\frac{n-3+2u}{16} \equiv 0 \pmod{p}$, 
we distinguish between the two subcases: $p$ odd and $p=2$.

If $p$ is odd, we have  
\begin{eqnarray*}
\lefteqn{m_\lambda(x)=} \\ 
& \left\{ \begin{array}{l} 
\frac{x^n-1}{(x-1)\Omega_3^{(4,n)}(x)  \Omega_2^{(4,n)}(x)}  \mbox{ if } \eta_0=1, \eta_1=0, \rho=0 \\  
\frac{x^n-1}{(x-1)\Omega_1^{(4,n)}(x)  \Omega_2^{(4,n)}(x)}  \mbox{ if } \eta_0=1, \eta_1=-1, \rho=0 \\ 
\frac{x^n-1}{(x-1)\Omega_3^{(4,n)}(x)  \Omega_0^{(4,n)}(x)}  \mbox{ if } \eta_0=-1, \eta_1=0, \rho=0 \\  
\frac{x^n-1}{(x-1)\Omega_1^{(4,n)}(x)  \Omega_0^{(4,n)}(x)}  \mbox{ if } \eta_0=-1, \eta_1=-1, \rho=0 \\ 
\frac{x^n-1}{\Omega_1^{(4,n)}(x)  }  \mbox{ if }  \eta_1=0, \rho=1 \\  
\frac{x^n-1}{\Omega_3^{(4,n)}(x)  }  \mbox{ if }  \eta_1=-1, \rho=1 
                               \end{array} 
                    \right.                                          
\end{eqnarray*} 
and   
\begin{eqnarray*} 
\ls_\lambda =\left\{ \begin{array}{ll} 
n-\frac{n+1}{2}  & \mbox{ if } \eta_0=1, \eta_1=0, \rho=0 \\  
n-\frac{n+1}{2}  & \mbox{ if } \eta_0=1, \eta_1=-1, \rho=0 \\ 
n-\frac{n+1}{2}  & \mbox{ if } \eta_0=-1, \eta_1=0, \rho=0 \\  
n-\frac{n+1}{2}  & \mbox{ if } \eta_0=-1, \eta_1=-1, \rho=0 \\ 
n-\frac{n-1}{4}  & \mbox{ if }  \eta_1=0, \rho=1 \\  
n-\frac{n-1}{4}  & \mbox{ if }  \eta_1=-1, \rho=1.  
                               \end{array} 
                    \right.                                          
\end{eqnarray*} 
In this subcase, the cyclic code $\calC_\lambda$  over $\gf(q)$ defined by the sequence $\lambda^{\infty}$ has 
the generator polynomial $m_\lambda(x)$ and parameters $\left[n, n-\ls_\lambda, d \right]$. 
In addition, if $\eta_1=0$ and  $\rho=1$ or  $\eta_1=-1$ and  $\rho=1$, 
the minimum weight $d$ of the code has the lower bound of Lemma \ref{lem-O4bound}, 
provided that $\ord_{n}(q)=(n-1)/4$. In the rest four cases, the code is a duadic code and 
the minimum odd-like weigh 
$d_{odd} \ge \sqrt{n}$.

If $p=2$, we have  
\begin{eqnarray*}
\lefteqn{m_\lambda(x)=} \\
& \left\{ \begin{array}{l} 
\frac{x^n-1}{(x-1)\Omega_3^{(4,n)}(x)  \Omega_0^{(4,n)}(x) \Omega_2^{(4,n)}(x)}  \mbox{ if }  \eta_1=0, \rho=0 \\  
\frac{x^n-1}{(x-1)\Omega_1^{(4,n)}(x)  \Omega_0^{(4,n)}(x) \Omega_2^{(4,n)}(x)}  \mbox{ if }  \eta_1=-1, \rho=0 \\ 
\frac{x^n-1}{\Omega_1^{(4,n)}(x)  }  \mbox{ if }  \eta_1=0, \rho=1 \\  
\frac{x^n-1}{\Omega_3^{(4,n)}(x)  }  \mbox{ if }  \eta_1=-1, \rho=1 
                               \end{array} 
                    \right.                                          
\end{eqnarray*} 
and   
\begin{eqnarray*}
\ls_\lambda =\left\{ \begin{array}{ll} 
n-\frac{3n+1}{4} & \mbox{ if }  \eta_1=0, \rho=0 \\  
n-\frac{3n+1}{4} & \mbox{ if }  \eta_1=-1, \rho=0 \\ 
n-\frac{n-1}{4} & \mbox{ if }  \eta_1=0, \rho=1 \\  
n-\frac{n-1}{4} & \mbox{ if }  \eta_1=-1, \rho=1 
                               \end{array} 
                    \right.                                          
\end{eqnarray*} 
In this subcase, the cyclic code $\calC_\lambda$  over $\gf(q)$ defined by the sequence $\lambda^{\infty}$ has 
the generator polynomial $m_\lambda(x)$ and parameters $\left[n, n-\ls_\lambda, d \right]$. 
In addition, if $\eta_1=0$ and  $\rho=1$ or  $\eta_1=-1$ and  $\rho=1$, 
the minimum weight $d$ of the code has the lower bound of Lemma \ref{lem-O4bound}, 
provided that $\ord_{n}(q)=(n-1)/4$.

\item When $\frac{n+1-2u}{16} \equiv 1 \pmod{p}$ and $\frac{n-3+2u}{16} \not\equiv 0 \pmod{p}$, we 
distinguish between the two cases: $p$ odd and $p=2$. 

If $p$ is odd, 
\begin{eqnarray*}
m_\lambda(x)=\left\{ \begin{array}{ll} 
\frac{x^n-1}{(x-1) \Omega_2^{(4,n)}(x)} & \mbox{ if }  \eta_0=1, \rho=0 \\  
\frac{x^n-1}{(x-1) \Omega_0^{(4,n)}(x)} & \mbox{ if }  \eta_0=-1, \rho=0 \\  
x^n-1 & \mbox{ if }  \rho=1 
                               \end{array} 
                    \right.                                          
\end{eqnarray*} 
and 
\begin{eqnarray*}
\ls_\lambda =\left\{ \begin{array}{ll} 
n-\frac{n+3}{4} & \mbox{ if }  \eta_0=1, \rho=0 \\  
n-\frac{n+3}{4} & \mbox{ if }  \eta_0=-1, \rho=0 \\  
n & \mbox{ if }  \rho=1. 
                               \end{array} 
                    \right.                                          
\end{eqnarray*} 
In this case, the cyclic code $\calC_\lambda$  over $\gf(q)$ defined by the sequence $\lambda^{\infty}$ has 
the generator polynomial $m_\lambda(x)$ and parameters $\left[n, n-\ls_\lambda, d \right]$. 
In addition, if $\eta_0=1$ and  $\rho=0$ or  $\eta_0=-1$ and  $\rho=0$, 
the minimum weight $d$ of the code has the lower bound of Lemma \ref{lem-O4bound2}, 
provided that $\ord_{n}(q)=(n-1)/4$.  

If $p=2$, 
\begin{eqnarray*}
m_\lambda(x)=\left\{ \begin{array}{ll} 
\frac{x^n-1}{(x-1) \Omega_0^{(4,n)}(x) \Omega_2^{(4,n)}(x)} & \mbox{ if }  \rho=0 \\  
x^n-1 & \mbox{ if }  \rho=1 
                               \end{array} 
                    \right.                                          
\end{eqnarray*} 
and 
\begin{eqnarray*}
\ls_\lambda =\left\{ \begin{array}{ll} 
n-\frac{n+1}{2} & \mbox{ if }  \rho=0 \\  
n & \mbox{ if }  \rho=1. 
                               \end{array} 
                    \right.                                          
\end{eqnarray*} 
In this subcase, the cyclic code $\calC_\lambda$  over $\gf(q)$ defined by the sequence $\lambda^{\infty}$ has 
the generator polynomial $m_\lambda(x)$ and parameters $\left[n, n-\ls_\lambda, d \right]$. 
Furthermore, the code is a quadratic residue code and hence $d \ge \sqrt{n}$ if $\rho =0$ \cite{MacSlo}.

\item When $\frac{n+1-2u}{16}  \not\equiv 0, 1 \pmod{p}$ and $\frac{n-3+2u}{16} \equiv 0 \pmod{p}$, 
\begin{eqnarray*}
m_\lambda(x)=\left\{ \begin{array}{ll} 
\frac{x^n-1}{(x-1)\Omega_3^{(4,n)}(x)  } & \mbox{ if }  \eta_1=0, \rho=0 \\  
\frac{x^n-1}{(x-1)\Omega_1^{(4,n)}(x)  } & \mbox{ if }  \eta_1=-1, \rho=0 \\ 
\frac{x^n-1}{\Omega_1^{(4,n)}(x)  } & \mbox{ if }  \eta_1=0, \rho=1 \\  
\frac{x^n-1}{\Omega_3^{(4,n)}(x)  } & \mbox{ if }  \eta_1=-1, \rho=1 
                               \end{array} 
                    \right.                                          
\end{eqnarray*} 
and 
\begin{eqnarray*}
\ls_\lambda =\left\{ \begin{array}{ll} 
n-\frac{n+3}{4} & \mbox{ if }  \eta_1=0, \rho=0 \\  
n-\frac{n+3}{4} & \mbox{ if }  \eta_1=-1, \rho=0 \\ 
n-\frac{n-1}{4} & \mbox{ if }  \eta_1=0, \rho=1 \\  
n-\frac{n-1}{4} & \mbox{ if }  \eta_1=-1, \rho=1.  
                               \end{array} 
                    \right.                                          
\end{eqnarray*} 
In this case, the cyclic code $\calC_\lambda$  over $\gf(q)$ defined by the sequence $\lambda^{\infty}$ has 
the generator polynomial $m_\lambda(x)$ and parameters $\left[n, n-\ls_\lambda, d \right]$. 
In addition, if $\eta_1=0$ and  $\rho=0$ or  $\eta_1=-1$ and  $\rho=0$, 
the minimum weight $d$ of the code has the lower bound of Lemma \ref{lem-O4bound2}, 
provided that $\ord_{n}(q)=(n-1)/4$.  If $\eta_1=0$ and  $\rho=1$ or  $\eta_1=-1$ and  $\rho=1$, 
the minimum weight $d$ of the code has the lower bound of Lemma \ref{lem-O4bound}, 
provided that $\ord_{n}(q)=(n-1)/4$.

\item When $\frac{n+1-2u}{16}  \not\equiv 0, 1 \pmod{p}$ and $\frac{n-3+2u}{16} \not\equiv 0 \pmod{p}$, 
\begin{eqnarray*}
m_\lambda(x)=\left\{ \begin{array}{ll} 
\frac{x^n-1}{x-1} & \mbox{ if }  \rho=0 \\  
x^n-1 & \mbox{ if }  \rho=1 
                               \end{array} 
                    \right.                                          
\end{eqnarray*} 
and 
\begin{eqnarray*}
\ls_\lambda =\left\{ \begin{array}{ll} 
n-1 & \mbox{ if }  \rho=0 \\  
n & \mbox{ if }  \rho=1. 
                               \end{array} 
                    \right.                                          
\end{eqnarray*} 
In this case, the cyclic code $\calC_\lambda$  over $\gf(q)$ defined by the sequence $\lambda^{\infty}$ has 
the generator polynomial $m_\lambda(x)$ and parameters $\left[n, n-\ls_\lambda, d \right]$, where  
$d=n$ if $\rho=0$.  

\end{enumerate}
\end{theorem}

\begin{proof} 
We prove the conclusions on the linear span and minimal polynomial of the sequence $\lambda^\infty$ 
for Case 1 only. The conclusions of other cases can be similarly proved. 

Since $n \equiv 1 \pmod{8}$, $-1 \in C_0^{(4,n)}$. By the definition of cyclotomic numbers, we have 
\begin{eqnarray*} 
\eta_\ell^2 
&=&  \left(\sum_{i \in C_{\ell}^{(4,n)}} \eta^i \right)^2 \\
&=& (\ell, \ell)_4 \eta_0 + (\ell+3, \ell+3)_4 \eta_1 + \\ 
& & (\ell+2, \ell+2)_4 \eta_2 + (\ell+1, \ell+1)_4 \eta_3+\frac{n-1}{4}.  
\end{eqnarray*}   

It then follows from Table \ref{tab-adset2} and the cyclotomic numbers of order 4 for the case 
$n \equiv 1 \pmod{8}$ that 
\begin{eqnarray*} 
\eta_0^2 &=& \frac{3n-1-2u}{16} -\frac{u+1}{2} \eta_0 +\frac{v}{2} \eta_1 -\frac{v}{2} \eta_3,  \\
\eta_1^2 &=& \frac{3n-1-2u}{16} -\frac{u+1}{2} \eta_1 +\frac{v}{2} \eta_2 -\frac{v}{2} \eta_0,  \\
\eta_2^2 &=& \frac{3n-1-2u}{16} -\frac{u+1}{2} \eta_2 +\frac{v}{2} \eta_3 -\frac{v}{2} \eta_1,  \\
\eta_3^2 &=& \frac{3n-1-2u}{16} -\frac{u+1}{2} \eta_3 +\frac{v}{2} \eta_0 -\frac{v}{2} \eta_2. 
\end{eqnarray*} 
Whence, 
\begin{eqnarray}\label{eqn-sept291}
\left\{ \begin{array}{l}
\eta_0^2 + \eta_2^2 = \frac{3n-1-2u}{8} -\frac{u+1}{2} (\eta_0 + \eta_2),  \\
\eta_1^2 + \eta_3^2 = \frac{3n-1-2u}{8} -\frac{u+1}{2} (\eta_1 + \eta_3).
\end{array} 
\right. 
\end{eqnarray}

Since $n \equiv 1 \pmod{8}$, $-1 \in C_0^{(4,n)}$. By the definition of cyclotomic numbers, we have 
\begin{eqnarray*} 
\eta_\ell \eta_{\ell+2}  
&=&  \sum_{i \in C_{\ell}^{(4,n)}}  \sum_{j \in C_{\ell+2}^{(4,n)}}  \eta^{i-j}   \\
&=& (\ell+2, \ell)_4 \eta_0 + (\ell+1, \ell+3)_4 \eta_1 + \\ 
& & (\ell, \ell+2)_4 \eta_2 + (\ell+3, \ell+1)_4 \eta_3.  
\end{eqnarray*}   
It then follows from the cyclotomic numbers of order 4 that 
\begin{eqnarray}\label{eqn-sept292}
\left\{ \begin{array}{l}
\eta_0 \eta_2 = -\frac{n+1-2u}{16} + \frac{u-1}{4} (\eta_0 + \eta_2),  \\
\eta_1 \eta_3 = -\frac{n+1-2u}{16} + \frac{u-1}{4} (\eta_1 + \eta_3). 
\end{array} 
\right. 
\end{eqnarray}

Since $\frac{n-1}{4} \bmod{p} =0$, 
\begin{eqnarray}\label{eqn-sept293}
\Lambda(1)=\rho. 
\end{eqnarray}

Recall that $\eta_0+\eta_2=0$ and $\eta_1+\eta_3=-1$.  In Case 1, by (\ref{eqn-sept291}) and 
(\ref{eqn-sept292}), we have 
$$ 
\eta_0=\eta_2=0, \ \eta_1(\eta_1+1)=\eta_3(\eta_3+1)=0. 
$$  

It then follows from (\ref{eqn-opendayall22b}) and (\ref{eqn-sept293}) that 
\begin{eqnarray*}
\lefteqn{\gcd(\Lambda(x), x^n-1)=} \\ 
& \left\{ \begin{array}{ll} 
(x-1)\Omega_3^{(4,n)}(x) & \mbox{ if } \eta_1=0 \mbox{ and } \rho=0 \\  
(x-1)\Omega_1^{(4,n)}(x) & \mbox{ if } \eta_1=-1 \mbox{ and } \rho=0 \\  
\Omega_1^{(4,n)}(x) \Omega_0^{(4,n)}(x) \Omega_2^{(4,n)}(x) & \mbox{ if } \eta_1=0 \mbox{ and } \rho=1 \\  
\Omega_3^{(4,n)}(x) \Omega_0^{(4,n)}(x) \Omega_2^{(4,n)}(x) & \mbox{ if } \eta_1=-1 \mbox{ and } \rho=1. \\  
                               \end{array} 
                    \right.                                          
\end{eqnarray*} 

The desired conclusions on the linear span and the minimal polynomial of the sequence $\lambda^\infty$ 
for Case 1 then follow from  Lemma \ref{lem-ls0}.  

The desired conclusions on the dimension and the generator polynomial of the code 
$\C_\lambda$ follow from the conclusions on the linear span and the minimal polynomial 
of the sequence $\lambda^\infty$ and the definition of the code $\C_\lambda$. 
The conclusion on the minimum weight for each case follows from Lemmas (\ref{lem-O4bound}) or  
(\ref{lem-O4bound2}), or the square-root bound on the minimum weight in quadratic residue codes, 
or the square-root bound on the minimum odd-like weight in duadic codes \cite{HPbook}.  
\end{proof} 

\begin{example} 
Let $(p, m, n)=(2,1,113)$.  Then $q=2 \in C_0^{(4,n)}$ and $n=u^2+4v^2=(-7)^2+4\times 4^2$. 
Then 
$$ 
\frac{n+1-2u}{16} \bmod{p} =0 \mbox{ and } \frac{n-3+2u}{16} \bmod{p} =0.  
$$ 
So this is Case 1. Let $\rho=1$. 
Then $\calC_\lambda$ is a $[113,84,8]$ cyclic code over $\gf(q)$ 
with generator polynomial 
\begin{eqnarray*} 
x^{29} + x^{27} + x^{26} + x^{22} + x^{21} + x^{18} + x^{16} + \\ 
x^{13} + x^{11} + x^8 +
    x^7 + x^3 + x^2 + 1
\end{eqnarray*} 
The best binary linear code known of length 113 and dimension 84 has minimum weight 10. 
\end{example}

\begin{example} \label{exam-12.5t} 
Let $(p, m, n)=(2,1,113)$.  Then $q=2 \in C_0^{(4,n)}$ and $n=u^2+4v^2=(-7)^2+4\times 4^2$. 
Then 
$$ 
\frac{n+1-2u}{16} \bmod{p} =0 \mbox{ and } \frac{n-3+2u}{16} \bmod{p} =0.  
$$ 
So this is Case 1. Let $\rho=0$. 
Then $\calC_\lambda$ is a $[113,29,28]$ cyclic code over $\gf(q)$ 
with generator polynomial 
\begin{eqnarray*} 
x^{84} + x^{82} + x^{81} + x^{80} + x^{76} + x^{75} + x^{74} + \\ 
x^{73} + x^{72} + x^{70} +  x^{68} + x^{66} + x^{65} + x^{64} + \\ 
    x^{63} + x^{62} + x^{60} + x^{59} + x^{58} + x^{57} + x^{56} + \\ 
    x^{55} + x^{53} + x^{47} + x^{46} + x^{43} + x^{42} + x^{41} + \\ 
    x^{38} + x^{37} + x^{31} + x^{29} + x^{28} + x^{27} + x^{26} + \\ 
    x^{25} + x^{24} + x^{22} + x^{21} + x^{20} + x^{19} + x^{18} + \\ 
    x^{16} +x^{14} + x^{12} + x^{11} + x^{10} + x^9 + x^8 + \\ 
    x^4 +  
    x^3 + x^2 + 1. 
\end{eqnarray*} 
The best binary linear code known of length 113 and dimension 29 has minimum weight 32.   
\end{example}

\begin{example}\label{exam-12.6t}  
Let $(p, m, n)=(2,2,41)$.  Then $q=4 \in C_0^{(4,n)}$ and $n=u^2+4v^2=5^2+4\times 1^2$. 
Then 
$$ 
\frac{n+1-2u}{16} \bmod{p} =0 \mbox{ and } \frac{n-3+2u}{16} \bmod{p} =1.  
$$ 
So this is Case 2. Let $\rho=0$. 
Then $\calC_\lambda$ is a $[41,1,41]$ cyclic code over $\gf(q)$ 
with generator polynomial $(x^{41}-1)/(x-1)$. 
\end{example} 

\begin{example} \label{exam-12.7t} 
Let $(p, m, n)=(2,2,41)$.  Then $q=4 \in C_0^{(4,n)}$ and $n=u^2+4v^2=5^2+4\times 1^2$. 
Then 
$$ 
\frac{n+1-2u}{16} \bmod{p} =0 \mbox{ and } \frac{n-3+2u}{16} \bmod{p} =1.  
$$ 
So this is Case 2. Let $\rho=1$. 
Then $\calC_\lambda$ is a $[41,20,10]$ cyclic code over $\gf(q)$ 
with generator polynomial 
\begin{eqnarray*} 
x^{21} + x^{19} + x^{18} + x^{16} + x^{15} + x^{14} + x^{12} + \\ 
x^9 + x^7 + x^6 +
    x^5 + x^3 + x^2 + 1.
\end{eqnarray*}
\end{example} 

\begin{example} 
Let $(p, m, n)=(2,1,73)$.  Then $q=2 \in C_0^{(4,n)}$ and $n=u^2+4v^2=(-3)^2+4\times 4^2$. 
Then 
$$ 
\frac{n+1-2u}{16} \bmod{p} =1 \mbox{ and } \frac{n-3+2u}{16} \bmod{p} =0.  
$$ 
So this is Case 3. Let $\rho=0$. 
Then $\calC_\lambda$ is a $[73,55,6]$ cyclic code over $\gf(q)$ 
with generator polynomial 
$$
x^{18} + x^{16} + x^{15} + x^{14} + x^{11} + x^{10} + x^9 + \\ 
x^8 + x^7 + x^4 + x^3
    + x^2 + 1.
$$
This may be the first known binary cyclic code with parameters $[73,55,6]$. Earlier, only a linear 
code with the same parameters was known. 
\end{example}

\begin{example}  
Let $(p, m, n)=(2,1,89)$.  Then $q=2 \in C_0^{(4,n)}$ and $n=u^2+4v^2=5^2+4\times 4^2$. 
Then 
$$ 
\frac{n+1-2u}{16} \bmod{p} =1 \mbox{ and } \frac{n-3+2u}{16} \bmod{p} =0.  
$$ 
So this is Case 3. Let $\rho=0$. 
Then $\calC_\lambda$ is a $[89,67,7]$ cyclic code over $\gf(q)$ 
with generator polynomial 
$$
x^{22} + x^{19} + x^{17} + x^{15} + x^{12} + x^{11} + \\ 
x^{10} + x^7 + x^5 + x^3 + 1. 
$$
The best linear code with length 89 and dimension 67 has minimum weight 8.  
This may be the first cyclic code known with parameters $[89,67,7]$. 
\end{example} 

\begin{example} \label{exam-12.10t} 
Let $(p, m, n)=(2,1,73)$.  Then $q=2 \in C_0^{(4,n)}$ and $n=u^2+4v^2=(-3)^2+4\times 4^2$. 
Then 
$$ 
\frac{n+1-2u}{16} \bmod{p} =1 \mbox{ and } \frac{n-3+2u}{16} \bmod{p} =0.  
$$ 
So this is Case 3. Let $\rho=1$. 
Then $\calC_\lambda$ is a $[73,18,24]$ cyclic code over $\gf(q)$ 
with generator polynomial 
\begin{eqnarray*} 
x^{55} + x^{53} + x^{52} + x^{47} + x^{43} + x^{41} + x^{40} + \\ 
x^{39} + x^{38} + x^{37} +
    x^{35} + x^{34} + x^{32} + x^{31} + \\
    x^{30} + x^{25} + x^{24} + x^{23} + x^{21} + x^{20} + x^{18} + \\
    x^{17} + x^{16} + x^{15} + x^{14} + x^{12} + x^8 + x^3 + x^2 + 1.
\end{eqnarray*} 
This may be the first known binary cyclic code with parameters $[73,18,24]$. Earlier, only a linear 
code with the same parameters was known. 
\end{example} 

\begin{example}\label{exam-12.11t}  
Let $(p, m, n)=(2,1,89)$.  Then $q=2 \in C_0^{(4,n)}$ and $n=u^2+4v^2=5^2+4\times 4^2$. 
Then 
$$ 
\frac{n+1-2u}{16} \bmod{p} =1 \mbox{ and } \frac{n-3+2u}{16} \bmod{p} =0.  
$$ 
So this is Case 3. Let $\rho=1$. 
Then $\calC_\lambda$ is a $[89,22,28]$ cyclic code over $\gf(q)$ 
with generator polynomial 
\begin{eqnarray*}
x^{67} + x^{64} + x^{62} + x^{61} + x^{60} + x^{58} + x^{53} + \\ 
x^{52} + x^{51} + x^{50} + x^{48} + x^{47} +  x^{45} +  x^{44} + \\ 
x^{41} +   x^{39} + x^{36} + x^{31} + x^{28} + x^{26} + x^{23} + \\ 
x^{22} + x^{20} + x^{19} + x^{17} + x^{16} + x^{15} + x^{14} + \\ 
x^9 + x^7 + x^6 + x^5 + x^3 + 1. 
\end{eqnarray*}
The best linear code with length 89 and dimension 22 has minimum weight 28.  
This may be the first cyclic code known with these parameters. 
\end{example} 

\begin{example} \label{exam-12.12t} 
Let $(p, m, n)=(2,2,17)$.  Then $q=4 \in C_0^{(4,n)}$ and $n=u^2+4v^2=1^2+4\times 2^2$. 
Then 
$$ 
\frac{n+1-2u}{16} \bmod{p} =1 \mbox{ and } \frac{n-3+2u}{16} \bmod{p} =1.  
$$ 
So this is Case 4. Let $\rho=0$. 
Then $\calC_\lambda$ is a $[17,9,5]$ cyclic code over $\gf(q)$ 
with generator polynomial 
\begin{eqnarray*}
x^8 + x^7 + x^6 + x^4 + x^2 + x + 1. 
\end{eqnarray*}
The best linear code with length 17 and dimension 9 has minimum weight 7.  
\end{example} 

\begin{remark} 
It was proved in \cite{Dingt} that $C_0^{(4,n)}$ is a $(n, (n-1)/4, (n-3)/16, (n-1)/2)$ almost difference set in 
$(\gf(n), +)$ when $n =5^2+4v^2$ or $n =(-3)^2+4v^2$. Examples \ref{exam-12.7t}, 
\ref{exam-12.10t}, and \ref{exam-12.11t} show that the cyclic codes defined by such almost 
difference sets are very good.   
\end{remark} 

\begin{remark} 
It was proved in \cite{DHL99} that $C_0^{(4,n)} \cup \{0\}$ is a $(n, (n+3)/4, (n-5)/16, (n-1)/2)$ almost difference set in 
$(\gf(n), +)$ when $n =1^2+4v^2$ or $n =(-7)^2+4v^2$. Examples \ref{exam-12.5t} and 
\ref{exam-12.12t} indicate that the cyclic code defined by such almost 
difference sets are very good.   
\end{remark}

When $n \equiv 5 \pmod{8}$, the linear span and minimal polynomial of the sequence $\lambda^{\infty}$ 
as well as the parameters of the code $\C_\lambda$ are given in the following theorem. 

\begin{theorem}\label{thm-bilegendre85}
Let $\frac{n-1}{4} \equiv 0 \pmod{p}$ and $q \in C_0^{(4,n)}$, and let $n \equiv 5 \pmod{8}$. 
let $\lambda^{\infty}$ be the sequence of period $n$ over $\gf(q)$ defined in (\ref{eqn-seqorder42}). 
As before, $n=u^2+4v^2$ with $u \equiv 1 \pmod{4}$.  
\begin{enumerate}

\item When $\frac{3n-1+2u}{16} \equiv 0 \pmod{p}$ and $\frac{3n+3-2u}{16} \equiv 0 \pmod{p}$, 
\begin{eqnarray*}
\lefteqn{m_\lambda(x)=} \\ 
&\left\{ \begin{array}{l} 
\frac{x^n-1}{(x-1)\Omega_3^{(4,n)}(x)}  \mbox{ if } \eta_1=0 \mbox{ and } \rho=0 \\  
\frac{x^n-1}{(x-1)\Omega_1^{(4,n)}(x)}  \mbox{ if } \eta_1=-1 \mbox{ and } \rho=0 \\  
\frac{x^n-1}{\Omega_1^{(4,n)}(x) \Omega_0^{(4,n)}(x) \Omega_2^{(4,n)}(x)}  \mbox{ if } \eta_1=0 \mbox{ and } \rho=1 \\  
\frac{x^n-1}{\Omega_3^{(4,n)}(x) \Omega_0^{(4,n)}(x) \Omega_2^{(4,n)}(x)}  \mbox{ if } \eta_1=-1 \mbox{ and } \rho=1 \\  
                               \end{array} 
                    \right.                                          
\end{eqnarray*} 
and   
\begin{eqnarray*}
\ls_\lambda=\left\{ \begin{array}{ll} 
n-\frac{n+3}{4}  & \mbox{ if } \eta_1=0 \mbox{ and } \rho=0 \\  
n-\frac{n+3}{4} & \mbox{ if } \eta_1=-1 \mbox{ and } \rho=0 \\  
n- \frac{3n-3}{4} & \mbox{ if } \eta_1=0 \mbox{ and } \rho=1 \\  
n- \frac{3n-3}{4}  & \mbox{ if } \eta_1=-1 \mbox{ and } \rho=1. \\  
                               \end{array} 
                    \right.                                          
\end{eqnarray*} 
In this case, the cyclic code $\calC_\lambda$  over $\gf(q)$ defined by the sequence $\lambda^{\infty}$ has 
the generator polynomial $m_\lambda(x)$ and parameters $\left[n, n-\ls_\lambda, d \right]$. 
In addition, if $\eta_1=0$ and  $\rho=0$ or $\eta_1=-1$ and  $\rho=0$, 
the minimum weight $d$ of the code has the lower bound of Lemma \ref{lem-O4bound2}, 
provided that $\ord_{n}(q)=(n-1)/4$.

\item  When $\frac{3n-1+2u}{16} \equiv 0 \pmod{p}$ and $\frac{3n+3-2u}{16} \not\equiv 0 \pmod{p}$, 
\begin{eqnarray*}
m_\lambda(x)=\left\{ \begin{array}{ll} 
\frac{x^n-1}{x-1} & \mbox{ if }  \rho=0 \\  
\frac{x^n-1}{\Omega_0^{(4,n)}(x) \Omega_2^{(4,n)}(x)} & \mbox{ if }  \rho=1  
                               \end{array} 
                    \right.                                          
\end{eqnarray*} 
and   
\begin{eqnarray*}
\ls_\lambda=\left\{ \begin{array}{ll} 
n-1  & \mbox{ if }  \rho=0 \\  
n- \frac{n-1}{2}  & \mbox{ if }  \rho=1. \\  
                               \end{array} 
                    \right.                                          
\end{eqnarray*} 
In this case, the cyclic code $\calC_\lambda$  over $\gf(q)$ defined by the sequence $\lambda^{\infty}$ has 
the generator polynomial $m_\lambda(x)$ and parameters $\left[n, n-\ls_\lambda, d \right]$, where 
\begin{eqnarray*}
\left\{ \begin{array}{ll} 
d=n  & \mbox{ if }  \rho=0 \\  
d \ge \sqrt{n}   & \mbox{ if }  \rho=1. \\  
                               \end{array} 
                    \right.                                          
\end{eqnarray*}

\item When $\frac{3n-1+2u}{16} \equiv p-1 \pmod{p}$ and $\frac{3n+3-2u}{16} \equiv 0 \pmod{p}$, 
\begin{eqnarray*}
\lefteqn{m_\lambda(x)=} \\ 
& \left\{ \begin{array}{l} 
\frac{x^n-1}{(x-1)\Omega_3^{(4,n)}(x)  \Omega_2^{(4,n)}(x)}  \mbox{ if } \eta_0=1, \eta_1=0, \rho=0 \\  
\frac{x^n-1}{(x-1)\Omega_1^{(4,n)}(x)  \Omega_2^{(4,n)}(x)}  \mbox{ if } \eta_0=1, \eta_1=-1, \rho=0 \\ 
\frac{x^n-1}{(x-1)\Omega_3^{(4,n)}(x)  \Omega_0^{(4,n)}(x)}  \mbox{ if } \eta_0=-1, \eta_1=0, \rho=0 \\  
\frac{x^n-1}{(x-1)\Omega_1^{(4,n)}(x)  \Omega_0^{(4,n)}(x)}  \mbox{ if } \eta_0= \eta_1=-1, \rho=0 \\ 
\frac{x^n-1}{\Omega_1^{(4,n)}(x)  }  \mbox{ if }  \eta_1=0, \rho=1 \\  
\frac{x^n-1}{\Omega_3^{(4,n)}(x)  }  \mbox{ if }  \eta_1=-1, \rho=1 
                               \end{array} 
                    \right.                                          
\end{eqnarray*} 
and   
\begin{eqnarray*} 
\ls_\lambda =\left\{ \begin{array}{ll} 
n-\frac{n+1}{2}  & \mbox{ if } \eta_0=1, \eta_1=0, \rho=0 \\  
n-\frac{n+1}{2}  & \mbox{ if } \eta_0=1, \eta_1=-1, \rho=0 \\ 
n-\frac{n+1}{2}  & \mbox{ if } \eta_0=-1, \eta_1=0, \rho=0 \\  
n-\frac{n+1}{2}  & \mbox{ if } \eta_0=-1, \eta_1=-1, \rho=0 \\ 
n-\frac{n-1}{4}  & \mbox{ if }  \eta_1=0, \rho=1 \\  
n-\frac{n-1}{4}  & \mbox{ if }  \eta_1=-1, \rho=1.  
                               \end{array} 
                    \right.                                          
\end{eqnarray*} 
In this case, the cyclic code $\calC_\lambda$  over $\gf(q)$ defined by the sequence $\lambda^{\infty}$ has 
the generator polynomial $m_\lambda(x)$ and parameters $\left[n, n-\ls_\lambda, d \right]$.  
In addition, if $\eta_1=0$ and  $\rho=1$ or  $\eta_1=-1$ and  $\rho=1$, 
the minimum weight $d$ of the code has the lower bound of Lemma \ref{lem-O4bound}, 
provided that $\ord_{n}(q)=(n-1)/4$. 
In the rest four cases, the code is a duadic code and hence the minimum odd-like weigh 
$d_{odd} \ge \sqrt{n}$.

\item When $\frac{3n-1+2u}{16} \equiv p-1 \pmod{p}$ and $\frac{3n+3-2u}{16} \not\equiv 0 \pmod{p}$, 
\begin{eqnarray*}
m_\lambda(x)=\left\{ \begin{array}{ll} 
\frac{x^n-1}{(x-1) \Omega_2^{(4,n)}(x)} & \mbox{ if }  \eta_0=1, \rho=0 \\  
\frac{x^n-1}{(x-1) \Omega_0^{(4,n)}(x)} & \mbox{ if }  \eta_0=-1, \rho=0 \\  
x^n-1 & \mbox{ if }  \rho=1 
                               \end{array} 
                    \right.                                          
\end{eqnarray*} 
and 
\begin{eqnarray*}
\ls_\lambda =\left\{ \begin{array}{ll} 
n-\frac{n+3}{4} & \mbox{ if }  \eta_0=1, \rho=0 \\  
n-\frac{n+3}{4} & \mbox{ if }  \eta_0=-1, \rho=0 \\  
n & \mbox{ if }  \rho=1. 
                               \end{array} 
                    \right.                                          
\end{eqnarray*} 
In this case, the cyclic code $\calC_\lambda$  over $\gf(q)$ defined by the sequence $\lambda^{\infty}$ has 
the generator polynomial $m_\lambda(x)$ and parameters $\left[n, n-\ls_\lambda, d \right]$.   
In addition, if $\eta_0=1$ and  $\rho=0$ or  $\eta_0=-1$ and  $\rho=0$, 
the minimum weight $d$ of the code has the lower bound of Lemma \ref{lem-O4bound2}, 
provided that $\ord_{n}(q)=(n-1)/4$.

\item When $\frac{3n-1+2u}{16} \not\equiv 0, p-1  \pmod{p} $ and $\frac{3n+3-2u}{16} \equiv 0 \pmod{p}$, 
\begin{eqnarray*}
m_\lambda(x)=\left\{ \begin{array}{ll} 
\frac{x^n-1}{(x-1)\Omega_3^{(4,n)}(x)  } & \mbox{ if }  \eta_1=0, \rho=0 \\  
\frac{x^n-1}{(x-1)\Omega_1^{(4,n)}(x)  } & \mbox{ if }  \eta_1=-1, \rho=0 \\ 
\frac{x^n-1}{\Omega_1^{(4,n)}(x)  } & \mbox{ if }  \eta_1=0, \rho=1 \\  
\frac{x^n-1}{\Omega_3^{(4,n)}(x)  } & \mbox{ if }  \eta_1=-1, \rho=1 
                               \end{array} 
                    \right.                                          
\end{eqnarray*} 
and 
\begin{eqnarray*}
\ls_\lambda =\left\{ \begin{array}{ll} 
n-\frac{n+3}{4} & \mbox{ if }  \eta_1=0, \rho=0 \\  
n-\frac{n+3}{4} & \mbox{ if }  \eta_1=-1, \rho=0 \\ 
n-\frac{n-1}{4} & \mbox{ if }  \eta_1=0, \rho=1 \\  
n-\frac{n-1}{4} & \mbox{ if }  \eta_1=-1, \rho=1.  
                               \end{array} 
                    \right.                                          
\end{eqnarray*} 
In this case, the cyclic code $\calC_\lambda$  over $\gf(q)$ defined by the sequence $\lambda^{\infty}$ has 
the generator polynomial $m_\lambda(x)$ and parameters $\left[n, n-\ls_\lambda, d \right]$. 
Furthermore, if $\eta_1=0$ and  $\rho=0$ or  $\eta_1=-1$ and  $\rho=0$, 
the minimum weight $d$ of the code has the lower bound of Lemma \ref{lem-O4bound2}, 
provided that $\ord_{n}(q)=(n-1)/4$.  If $\eta_1=0$ and  $\rho=1$ or  $\eta_1=-1$ and  $\rho=1$, 
the minimum weight $d$ of the code has the lower bound of Lemma \ref{lem-O4bound}, 
provided that $\ord_{n}(q)=(n-1)/4$.

\item When $\frac{3n-1+2u}{16} \not\equiv 0, p-1  \pmod{p} $  and $\frac{3n+3-2u}{16} \not\equiv 0 \pmod{p}$, 
\begin{eqnarray*}
m_\lambda(x)=\left\{ \begin{array}{ll} 
\frac{x^n-1}{x-1} & \mbox{ if }  \rho=0 \\  
x^n-1 & \mbox{ if }  \rho=1 
                               \end{array} 
                    \right.                                          
\end{eqnarray*} 
and 
\begin{eqnarray*}
\ls_\lambda =\left\{ \begin{array}{ll} 
n-1 & \mbox{ if }  \rho=0 \\  
n & \mbox{ if }  \rho=1. 
                               \end{array} 
                    \right.                                          
\end{eqnarray*} 
In this case, the cyclic code $\calC_\lambda$  over $\gf(q)$ defined by the sequence $\lambda^{\infty}$ has 
the generator polynomial $m_\lambda(x)$ and parameters $\left[n, n-\ls_\lambda, d \right]$, where  
$d=n$ if $\rho=0$. 
\end{enumerate}
\end{theorem}

\begin{proof} 
We prove only the conclusion of Case 3. The conclusions of other cases can be similarly proved. 

Since $n \equiv 5 \pmod{8}$, $-1 \in C_2^{(4,n)}$. Note that $p$ must be odd, as $n \equiv 5 \pmod{8}$ 
and $\frac{n-1}{4} \bmod{p}=0$. By the definition of cyclotomic numbers, we have 
\begin{eqnarray*} 
\eta_\ell^2 
&=&  \left(\sum_{i \in C_{\ell}^{(4,n)}} \eta^i \right)^2 \\
&=&  \sum_{i \in C_{\ell}^{(4,n)}}  \sum_{j \in C_{\ell +2}^{(4,n)}}  \eta^{i-j}  \\
&=& (\ell+2, \ell)_4 \eta_0 + (\ell+1, \ell+3)_4 \eta_1 + \\ 
& & (\ell, \ell+2)_4 \eta_2 + (\ell+3, \ell+1)_4 \eta_3.  
\end{eqnarray*}   

It then follows from Table \ref{tab-adset1} and the cyclotomic numbers of order 4 for the case 
$n \equiv 1 \pmod{8}$ that 
\begin{eqnarray*} 
\eta_0^2 &=& -\frac{n+1-6u}{16} + \frac{u-1}{2} \eta_0 +\frac{u-v}{2} \eta_1 -\frac{u+v}{2} \eta_3,  \\
\eta_1^2 &=& -\frac{n+1-6u}{16} + \frac{u-1}{2} \eta_1 +\frac{u-v}{2} \eta_2 -\frac{u+v}{2} \eta_0,  \\
\eta_2^2 &=& -\frac{n+1-6u}{16} + \frac{u-1}{2} \eta_2 +\frac{u-v}{2} \eta_3 -\frac{u+v}{2} \eta_1,  \\
\eta_3^2 &=& -\frac{n+1-6u}{16} + \frac{u-1}{2} \eta_3 +\frac{u-v}{2} \eta_0 -\frac{u+v}{2} \eta_2. 
\end{eqnarray*} 
Whence, 
\begin{eqnarray}\label{eqn-sept291b}
\left\{ \begin{array}{l}
\eta_0^2 + \eta_2^2 = -\frac{n+1+2u}{8} -\frac{u+1}{2} (\eta_0 + \eta_2),  \\
\eta_1^2 + \eta_3^2 = -\frac{n+1+2u}{8} -\frac{u+1}{2} (\eta_1 + \eta_3).
\end{array} 
\right. 
\end{eqnarray}

Since $n \equiv 1 \pmod{8}$, $-1 \in C_2^{(4,n)}$. By the definition of cyclotomic numbers, we have 
\begin{eqnarray*} 
\eta_\ell \eta_{\ell+2}  
&=&  \sum_{i \in C_{\ell}^{(4,n)}}  \sum_{j \in C_{\ell}^{(4,n)}}  \eta^{i-j}   \\
&=& (\ell, \ell)_4 \eta_0 + (\ell+3, \ell+3)_4 \eta_1 + \\ 
& & (\ell+2, \ell+2)_4 \eta_2 + (\ell+1, \ell+1)_4 \eta_3 + \frac{n-1}{4}.  
\end{eqnarray*}   
It then follows from the cyclotomic numbers of order 4 that 
\begin{eqnarray}\label{eqn-sept292b}
\left\{ \begin{array}{l}
\eta_0 \eta_2 = \frac{3n-1+2u}{16} + \frac{u-1}{4} (\eta_0 + \eta_2),  \\
\eta_1 \eta_3 = \frac{3n-1+2u}{16} + \frac{u-1}{4} (\eta_1 + \eta_3). 
\end{array} 
\right. 
\end{eqnarray}

Since $\frac{n-1}{4} \bmod{p} =0$, 
\begin{eqnarray}\label{eqn-sept293b}
\Lambda(1)=\rho. 
\end{eqnarray}

Recall that $\eta_0+\eta_2=0$ and $\eta_1+\eta_3=-1$.  In Case 3, by (\ref{eqn-sept291b}) and 
(\ref{eqn-sept292b}), we have 
$$ 
\eta_0=\eta_2=1, \ \eta_1(\eta_1+1)=\eta_3(\eta_3+1)=0. 
$$  

It then follows from (\ref{eqn-opendayall22b}) and (\ref{eqn-sept293b}) that 
\begin{eqnarray*}
\lefteqn{\gcd(\Lambda(x), x^n-1)=} \\ 
& 
\left\{ \begin{array}{ll} 
(x-1)\Omega_3^{(4,n)}(x)  \Omega_2^{(4,n)}(x) & \mbox{ if } \eta_0=1, \eta_1=0, \rho=0 \\  
(x-1)\Omega_1^{(4,n)}(x)  \Omega_2^{(4,n)}(x) & \mbox{ if } \eta_0=1, \eta_1=-1, \rho=0 \\ 
(x-1)\Omega_3^{(4,n)}(x)  \Omega_0^{(4,n)}(x) & \mbox{ if } \eta_0=-1, \eta_1=0, \rho=0 \\  
(x-1)\Omega_1^{(4,n)}(x)  \Omega_0^{(4,n)}(x) & \mbox{ if } \eta_0= \eta_1=-1, \rho=0 \\ 
\Omega_1^{(4,n)}(x)   & \mbox{ if }  \eta_1=0, \rho=1 \\  
\Omega_3^{(4,n)}(x)   & \mbox{ if }  \eta_1=-1, \rho=1 
                               \end{array} 
                    \right.                                          
\end{eqnarray*} 

The desired conclusions on the linear span and the minimal polynomial of the sequence $\lambda^\infty$ 
for Case 3 then follow from  Lemma \ref{lem-ls0}.  

The desired conclusions on the dimension and the generator polynomial of the code 
$\C_\lambda$ follow from the conclusions on the linear span and the minimal polynomial 
of the sequence $\lambda^\infty$ and the definition of the code $\C_\lambda$. 
The conclusion on the minimum weight for each case follows from Lemmas (\ref{lem-O4bound}) or  
(\ref{lem-O4bound2}), or the square-root bound on the minimum weight in quadratic residue codes, 
or the square-root bound on the minimum odd-like weight in duadic codes \cite{HPbook}.  
\end{proof}

\begin{example} \label{exam-12.14t}
Let $(p, m, n)=(3,2,61)$.  Then $q=9 \in C_0^{(4,n)}$ and $n=u^2+4v^2=5^2+4\times 3^2$. 
Then 
$$ 
\frac{3n-1+2u}{16} \bmod{p} = 0 \mbox{ and } \frac{3n+3-2u}{16} \bmod{p} =2.  
$$ 
So this is Case 2. Let $\rho=1$. 
Then $\calC_\lambda$ is a $[61,30,12]$ cyclic code over $\gf(q)$ 
with generator polynomial 
\begin{eqnarray*}
x^{31} + x^{29} + 2x^{28} + 2x^{27} + 2x^{26} + x^{25} + 2x^{22} + x^{19} + \\ 
    2x^{16} + x^{15} + 2x^{12} + x^9 + 2x^6 + x^5 + x^4 + x^3 + 2x^2 + 2. 
\end{eqnarray*}
The best linear code over $\gf(q)$ with length 61 and dimension 30 has minimum weight 20.  
\end{example}

\begin{example} \label{exam-12.15t} 
Let $(p, m, n)=(3,1,13)$.  Then $q=3 \in C_0^{(4,n)}$ and $n=u^2+4v^2=(-3)^2+4\times 1^2$. 
Then 
$$ 
\frac{3n-1+2u}{16} \bmod{p} = 2 \mbox{ and } \frac{3n+3-2u}{16} \bmod{p} =0.  
$$ 
So this is Case 3. Let $\rho=0$. 
Then $\calC_\lambda$ is a $[13,7,4]$ cyclic code over $\gf(q)$ 
with generator polynomial 
\begin{eqnarray*}
x^6 + 2x^5 + x^4 + 2x^3 + 2x^2 + 2x + 1. 
\end{eqnarray*}
The optimal linear code over $\gf(q)$ with length 13 and dimension 7 has minimum weight 5. 
The code of this example is almost optimal and cyclic.   
\end{example} 

\begin{example}  \label{exam-12.16t}
Let $(p, m, n)=(3,1,13)$.  Then $q=3 \in C_0^{(4,n)}$ and $n=u^2+4v^2=(-3)^2+4\times 1^2$. 
Then 
$$ 
\frac{3n-1+2u}{16} \bmod{p} = 2 \mbox{ and } \frac{3n+3-2u}{16} \bmod{p} =0.  
$$ 
So this is Case 3. Let $\rho=1$. 
Then $\calC_\lambda$ is a $[13,3,9]$ cyclic code over $\gf(q)$ 
with generator polynomial 
\begin{eqnarray*}
x^{10} + x^8 + x^7 + x^6 + 2x^5 + 2x^4 + x^2 + 2x + 1. 
\end{eqnarray*}
The known optimal linear code over $\gf(q)$ with length 13 and dimension 3 has minimum weight 9. 
The code of this example is both optimal and cyclic.   
\end{example}

\begin{example} \label{exam-12.17t} 
Let $(p, m, n)=(3,1,109)$.  Then $q=3 \in C_0^{(4,n)}$ and $n=u^2+4v^2=(-3)^2+4\times 5^2$. 
Then 
$$ 
\frac{3n-1+2u}{16} \bmod{p} = 2 \mbox{ and } \frac{3n+3-2u}{16} \bmod{p} =0.  
$$ 
So this is Case 3. Let $\rho=1$. 
Then $\calC_\lambda$ is a $[109,27,42]$ cyclic code over $\gf(q)$ 
with generator polynomial 
\begin{eqnarray*}
x^{82} + 2x^{80} + 2x^{79} + x^{78} + x^{77} + 2x^{76} + 2x^{75} + x^{74} + \\ 
x^{73} + 2x^{72} + 2x^{70} + 2x^{69} + 2x^{66} + 2x^{65} + 2x^{64} + 2x^{63} + \\ 
2x^{62} + 2x^{58} + x^{57} + x^{56} + 2x^{55} + x^{53} +  
 2x^{52} + 2x^{51} +  \\ 
    2x^{50} + x^{49} +  
    2x^{48} + 2x^{46} + 2x^{45} + x^{44} +  
    2x^{42} + x^{40} + \\ 
    2x^{39} + 2x^{38} + 2x^{35} +  
    2x^{32} + x^{31} + x^{29} + 2x^{28} + 2x^{27} + \\ x^{26} +  x^{25} + x^{24} +  
   x^{22} + x^{21} +
    x^{20} + 2x^{16} + 2x^{15} + \\ 
    2x^{14} + x^{12} +  x^{11} + 2x^8 + x^7 + 2x^5 + x^3 +
    2x + 1. 
\end{eqnarray*}
This has the same parameters as the best known code with length 109 and dimension 27 which is also cyclic.   
\end{example} 

\begin{example} 
Let $(p, m, n)=(7,1,29)$.  Then $q=7 \in C_0^{(4,n)}$ and $n=u^2+4v^2=5^2+4\times 1^2$. 
Then 
$$ 
\frac{3n-1+2u}{16} \bmod{p} = 6 \mbox{ and } \frac{3n+3-2u}{16} \bmod{p} =5.  
$$ 
So this is Case 4. Let $\rho=0$. 
Then $\calC_\lambda$ is a $[29,8,15]$ cyclic code over $\gf(q)$ 
with generator polynomial 
\begin{eqnarray*}
x^{21} + 3x^{19} + 2x^{18} + 5x^{17} + 5x^{16} + 6x^{15} + 5x^{14} + \\
4x^{13} + 4x^{12} +  x^{11} + 3x^{10} + x^9 + 4x^8 + 5x^7 +  \\ 
x^6 + x^5 + 6x^4 + 3x^3 +
    4x^2 + 5x + 6. 
\end{eqnarray*}
The known optimal linear code over $\gf(q)$ with length 29 and dimension 8 has minimum weight 17. 
\end{example}

\begin{remark} 
It was proved in \cite{Dingt} that $C_0^{(4,n)}$ is an $(n, (n-1)/4, (n-3)/16, (n-1)/2)$ almost difference set in 
$(\gf(n), +)$ when $n =(-3)^2+4v^2$ or $n =5^2+4v^2$. Examples \ref{exam-12.14t},   
\ref{exam-12.16t} and \ref{exam-12.17t} show that the cyclic code defined by such almost 
difference sets are very good.   
\end{remark} 

\begin{remark} 
It is known that $C_0^{(4,n)} \cup \{0\}$ is an $(n, (n-1)/4, (n+3)/16)$ difference set in 
$(\gf(n), +)$ when  $n =(-3)^2+4v^2$ and $v$ is odd. Examples \ref{exam-12.5t}  may 
indicate that the cyclic code defined by such 
difference sets are very good.   
\end{remark} 

\begin{open} 
Determine the parameters of the code $\C_\lambda$ defined by the sequence $\lambda^{\infty}$ of 
(\ref{eqn-seqorder42}) for the case that $\frac{n-1}{4} \not\equiv 0 \pmod{p}$. 
\end{open}

\section{Concluding remarks}

Perfect difference sets were used to construct cyclic codes in \cite{Weld}. The idea of 
constructing cyclic codes with special types of sequences employed in this paper could 
be viewed as an extension of this idea. 

There are several bounds on cyclic codes \cite{AF96,BS06,Bost,Bose,HT72,LintW}. It may 
not be easy to employ them to get tight bounds on the minimum weight of the cyclic 
codes presented in this paper. The actual miminum weight of these codes depends on 
the distribution of biquadratic residues modulo $n$, which looks to be a hard problem. 
However, some of the codes obtained in this paper are quadratic residue codes and duadic 
codes, which have a square-root bound on the minimum weight and the minimum odd-like 
weight respectively. In addition, we developed lower bounds on the minimum weight $d$ 
of some cyclic codes under certain conditions.  It would be nice if tight lower bounds on 
the minimum weight could be developed for the remaining cases.   

As a subclass of linear codes, cyclic codes usually have a smaller minimum weight 
compared with linear codes of the same length and dimension. However, some cyclic 
codes are optimal in the sense that they meet bounds defined for all linear codes, 
For example, the cyclic code of Example \ref{exam-order4balance1} is optimal. It 
is interesting to note that many of the example codes presented in this paper are 
the best possible cyclic codes and some are as good as the best linear codes with 
the same length and dimension.  For example, the binary cyclic code of Example 
\ref{exam-12.11t} has parameters $[89,22,28]$, which has the same parameters 
as the record binary linear code. These examples demonstrate that the cyclic codes 
defined by the cyclotomic sequences of order four are very good in general, but  
could be bad sometimes.    

The $p$-rank of the almost difference sets and difference sets is defined to be 
the linear span of the sequences over $\gf(p)$ defined by the almost difference 
sets and difference sets. The $p$-ranks of  the almost difference sets and 
difference sets can be used to distinguish them from other almost difference sets 
and difference sets. This is the contribution of this paper to combinatorics. 
The contribution of this paper to the theory of sequences and cryptography 
is the computation of the linear span of these cyclotomic sequences of order four.

\end{document}